\documentclass[12pt]{iopart}
\eqnobysec
\expandafter\let\csname equation*\endcsname\relax
\expandafter\let\csname endequation*\endcsname\relax
\usepackage{amsmath, amsthm, amssymb}
\usepackage{enumerate}
\usepackage{setspace}
\usepackage{natbib}

\newtheorem{theorem}{Theorem}[section]
\newtheorem{lemma}[theorem]{Lemma}
\newtheorem{proposition}[theorem]{Proposition}

\newenvironment{definition}[1][Definition]{\begin{trivlist}
\item[\hskip \labelsep {\bfseries #1}]}{\end{trivlist}}
\newenvironment{example}[1][Example]{\begin{trivlist}
\item[\hskip \labelsep {\bfseries #1}]}{\end{trivlist}}
\newenvironment{remark}[1][Remark]{\begin{trivlist}
\item[\hskip \labelsep {\bfseries #1}]}{\end{trivlist}}

\newcommand{\derv}[1]{\frac{\partial}{\partial #1}}

\newcommand{\deriv}[2]{\frac{\partial #1}{\partial #2}}

\newcommand{\beqn}{\begin{equation}}
\newcommand{\eeqn}{\end{equation}}
\newcommand{\beqnar}{\begin{eqnarray}}
\newcommand{\eeqnar}{\end{eqnarray}}

\newcommand{\contr}{\,\lrcorner\,}

\begin{document}

\title[Advected Invariants in MHD and Gas Dynamics]
{Local and Nonlocal Advected Invariants and Helicities in Magnetohydrodynamics 
and Gas Dynamics I: Lie Dragging Approach}
\author{G.M. Webb${}^1$, B. Dasgupta${}^1$, J. F. McKenzie${}^{1,3}$, 
Q. Hu${}^{1,2}$  and
G. P. Zank${}^{1,2}$}
\address{${}^1$ CSPAR, The University of Alabama Huntsville in
 Huntsville AL 35805, USA}

\address{${}^2$ Dept. of Physics, The University of Alabama in Huntsville, Huntsville 
AL 35899, USA}

\address{${}^3$Department of Mathematics and Statistics,\\
Durban University of Technology,\\ Steve Biko Campus, Durban South Africa\\
and School of Mathematical Sciences,\\
University of Kwa-Zulu, Natal, Durban, South Africa}
 
\ead{gmw0002@uah.edu}


\begin{abstract}
In this paper advected invariants and conservation laws in ideal 
magnetohydrodynamics (MHD) and gas dynamics are obtained using 
Lie dragging techniques. 
There are different classes of invariants that are advected or 
Lie dragged with the flow. 
  Simple examples are the advection of the entropy 
S (a 0-form), and the conservation of magnetic flux (an invariant 2-form advected with the flow). 
The magnetic flux conservation law is equivalent to Faraday's equation.   
The gauge condition for the magnetic helicity to be advected with the flow 
is determined.   
Different variants of the helicity in ideal fluid dynamics and MHD 
including:  fluid helicity, cross helicity and magnetic helicity 
are investigated.  
The fluid helicity conservation law and the cross helcity conservation 
law in MHD are derived for the case of a barotropic gas. If 
the magnetic field lies in the constant entropy surface, then the gas 
pressure can depend on both the entropy and the density. In 
these cases the conservation laws are local conservation laws. 
For non-barotropic gases, we obtain nonlocal conservation laws 
for fluid helicity and cross 
helicity by using  Clebsch variables. These nonlocal conservation laws 
are the main new results of the paper.   
Ertel's theorem and potential vorticity, the Hollman invariant,  
and the Godbillon Vey invariant for special  flows for which 
the magnetic helicity is zero are also discussed. 
\end{abstract}

\submitto{{\it J. Phys. A., Math. and Theor.},19th June 2013,
 Revised \today}

\maketitle


\section{Introduction}

Advected invariants and conservation laws in ideal magnetohydrodynamics and gas dynamics have wide applications in space plasma physics, fusion and laboratory 
plasmas, and fluid dynamics. In space plasma physics and solar physics, 
magnetic helicity is of major interest in describing the topology and 
linkage of magnetic fields (e.g. Berger and Field (1984), 
Moffatt (1969), Moffatt and Ricca (1992)
,Woltjer (1958),
Berger (1999), Berger and Ruzmaikin (2000),
 Bieber et al. (1987), Yahalom and Lynden Bell (2008), Yahalom (2013),
Webb et al. (2010a,b)). Ruzmaikin and Akhmetiev (1994), Akhmetiev and 
Ruzmaikin (1995) and Berger (1990) investigated 
higher order knot invariants in MHD, 
known as Sato-Levine invariants. These invariants are 
associated with 2-sided surfaces known as Seifert surfaces 
in which the knot is imbedded. They 
illustrated the theory using the Whitehead link and  Borromean rings.  
Kuznetsov and Ruban (1998,2000) and Kuznetsov (2006) have developed 
the Hamiltonian dynamics of vortex and magnetic field lines 
in hydrodynamic type systems. They use a mixed Eulerian and Lagrangian 
description (the so-called vortex line representation (VLR)). The VLR mapping 
describes the compressiblity of the vortex lines even in incompressible 
flows, and can be used to describe the merging and collapse of the vortex 
lines. This work is clearly important in the development of 
topological fluid dynamics and Hamiltonian fluid dynamics, 
 but is not explicitly addressed 
in the present paper.

The present paper gives a synposis of advected invariants in 
magnetohydrodynamics (MHD) and gas dynamics.  A short 
account of this work is given by Webb et al. (2013) in a preliminary 
conference paper. 
The discussion is based in part on the paper 
of Tur and Yanovsky (1993), who use the ideas of Lie dragging of 
vectors, $n$-forms ($n=1,2,3$), scalars and tensors  
(i.e. conserved physical quantities), and the algebra of exterior differential 
forms to determine the advected invariants.

The concept of Lie dragging of advected invariants in MHD and gas dynamics 
was investigated by Tur and Janovsky (1993) who extended previous work 
by Moiseev et al. (1982). These ideas also appear in recent work on advected 
invariants and conservation laws in MHD and hydrodynamical models 
by Cotter and Holm (2012). Cotter and Holm (2012) use an approach based on the 
Eulerian, Euler Poincar\'e formulation of ideal hydrodynamical models. 
Cotter and Holm (2012) derive conservation laws associated with 
Noether's second theorem, due to fluid relabeling symmetries. Their analysis 
does not require the introduction of Lagrangian variables. However, 
to obtain some of the conservation laws  
 it is necessary to include Lagrange multipliers to take into 
account constraints in their variational principle.  
Hydon and Mansfield (2011)  
 discuss and extend Noether's second theorem by using Lagrange multipliers 
which explicitly shows that for variational problems  
involving free functions and infinite dimensional Lie algebraic structures, 
that there are differential relations between the different Euler operators 
occuring in Noether's second theorem that must be taken into account 
(in Padhye and Morrison (1996a,b) these relations are referred to as 
generalized Bianchi identities). The work of Cotter and Holm (2012) extends
previous techniques used to derive conservation laws due to fluid 
relabelling symmetries. 
 Earlier work by Salmon (1982,1988) and 
Padhye and Morrison (1996a,b) used a Lagrangian fluid dynamics approach. 

Section 2 outlines the model equations. 

Section 3 gives an overview of helicity in ideal fluid mechanics and MHD. The local helicity conservation law in 
ideal fluid mechanics is given for the case of an isobaric equation of 
state for the gas (i.e. the pressure $p=p(\rho)$ is the equation of state for 
the gas). Integral forms of the helicity conservation equation and Ertel's 
theorem  in ideal fluid mechanics are discussed. Conservation laws for 
magnetic helicity and cross helicity in MHD are described. The concept of 
relative helicity in MHD is also described. 

Section 4 outlines a theory for 
advected invariants in ideal fluid mechanics and MHD, based on the Lie 
dragging of invariant geometrical quantities with the fluid (e.g. 
vector fields and $p$-forms where $p=0,1,2,3$ in 3D MHD). The discussion 
is based on the work of Tur and Janovsky (1993) on advected invariants 
in fluid and MHD systems of equations.  We discuss the 
concept of topological charge for invariant advected differential forms 
in fluid dynamics and MHD. We also derive and discuss the Godbillon-Vey 
topological invariant. The Godbillon Vey invariant 
in an MHD flow, arises for example, if 
${\bf A}{\bf\cdot}\nabla\times {\bf A}=0$
where ${\bf A}$ is the magnetic vector potential 
and ${\bf B}=\nabla\times{\bf A}$ 
is the magnetic induction. In such a flow, the magnetic helicity is zero. 
However, there is a higher order topological invariant (the Godbillon Vey 
invariant), which in general is non-zero. Thus a zero magnetic helicity 
field, can still have a non-trivial topology.  The Godbillon Vey invariant 
also occurs in ideal fluid dynamics for flows in which the fluid helicity 
${\bf u}{\bf\cdot}\nabla\times{\bf u}=0$.  

Section 5 gives an overview 
of the use of Clebsch variables in Lagrangian and Hamiltonian fluid mechanics. 
We also discuss the canonical and non-canonical Poisson bracket for MHD 
(e.g. Morrison and Greene (1980,1982), Holm and Kupershmidt (1983a,b)) 
and Weber transformations. 

Section 6 uses a Clebsch variable 
formulation of ideal fluid mechanics to derive a nonlocal helicity 
conservation law (\ref{eq:nl1}) for a fluid with a non-barotropic 
equation of state (i.e. $p=p(\rho,S)$). Clebsch variables are also used 
to derive the nonlocal cross helicity conservation law (\ref{eq:gch1}) 
in MHD. These nonlocal conservation laws for helicity and cross helicity 
are two new results obtained in the present paper. 

Section 7 concludes with a  summary and discussion.

\section{The Model}

The magnetohydrodynamic equations can be written in the form:

\begin{eqnarray}
&&\deriv{\rho}{t}+\nabla{\bf\cdot} (\rho{\bf u})=0, \label{eq:2.1}\\
&&\derv{t}(\rho{\bf u})+\nabla{\bf\cdot}\left[\rho {\bf u}{\bf u}
+\left(p+\frac{B^2}{2\mu}\right){\bf I} -\frac{\bf B B}{\mu}\right]=0,
\label{eq:2.2}\\
&&\deriv{S}{t}+{\bf u}{\bf\cdot}\nabla S=0, 
\label{eq:2.3}\\
&&\deriv{\bf B}{t}-\nabla\times\left({\bf u}\times {\bf B}\right)
+{\bf u}\nabla{\bf\cdot}{\bf B}=0, \label{eq:2.4}
\end{eqnarray}
where $\rho$, ${\bf u}$, $p$,
$S$ and ${\bf B}$ correspond to the gas density, fluid velocity, pressure,
specific entropy, and magnetic induction ${\bf B}$ respectively, and
${\bf I}$ is the unit $3\times 3$ dyadic.
The gas pressure $p=p(\rho,S)$ is a function of the density $\rho$ and
entropy $S$, and $\mu$ is the magnetic permeability. 
Equations (\ref{eq:2.1})-(\ref{eq:2.2}) are the mass and  momentum  
conservation laws, (\ref{eq:2.3}) is the entropy advection equation 
and (\ref{eq:2.4}) is Faraday's equation in the MHD limit. 

In classical MHD, (\ref{eq:2.1})-(\ref{eq:2.4}) are supplemented by
Gauss' law:
\beqn
\nabla{\bf\cdot}{\bf B}=0. \label{eq:2.5}
\eeqn
which implies the non-existence of magnetic monopoles.

It is useful to keep in mind the first law of thermodynamics:
\beqn
TdS=dQ=dU+pdV\quad\hbox{where}\quad V=\frac{1}{\rho}, \label{eq:2.7}
\eeqn
where $U$ is the internal energy per unit mass and $V=1/\rho$ is the specific
volume. To be more exact, the first law of thermodynamics should be  
 $\delta Q=dU+\delta W$, 
in order to emphasize that in general, the internal 
energy density is a perfect differential, whereas $\delta Q$ and $\delta W$
are not necessarily perfect differentials (e.g. for the case of heat 
conduction and dissipation $\delta Q$ is not a perfect differential). 
 However, for ideal 
gases, with no dissipation $\delta Q=TdS$ where $T$ is the temperature 
and $dS$ is a perfect differential. The integrability of 
Pfaffian differential forms of the form $Xdx+Ydy+Zdz=0$ is only 
assured iff ${\bf X}{\bf\cdot}\nabla\times{\bf X}=0$ where 
${\bf X}=(X,Y,Z)$. This result plays an important role in 
Caratheodory thermodynamics (e.g Sneddon (1957)). 
Using the internal energy per unit volume $\varepsilon=\rho U$
instead of $U$, (\ref{eq:2.7}) may be written as:
\beqn
TdS=\frac{1}{\rho}\left(d\varepsilon-hd\rho\right)\quad\hbox{where}\quad
h=\frac{\varepsilon+p}{\rho}, \label{eq:2.8}
\eeqn
is the enthalpy of the gas. Assuming $\varepsilon=\varepsilon(\rho,S)$
 (\ref{eq:2.8}) gives  the formulae:
\beqn
 \rho T=\varepsilon_S, \quad h=\varepsilon_\rho,
\quad p=\rho\varepsilon_{\rho}-\varepsilon, \label{eq:2.9}
\eeqn
relating the temperature $T$, enthalpy $h$ and pressure $p$ to the
internal energy density $\varepsilon(\rho,S)$. From (\ref{eq:2.8})
we obtain:
\beqn
TdS=dh-\frac{1}{\rho}dp\quad \hbox{and}\quad
-\frac{1}{\rho}\nabla p=T\nabla S-\nabla h, \label{eq:2.10}
\eeqn
which is useful in the further analysis of the momentum equation for the
system.

It is worth noting that the different variants of helicity (fluid 
 helicity, cross helicity and magnetic helicity) are pseudo-scalars 
as they reverse sign under space reversal ${\bf x}\to -{\bf x}$ 
(e.g. ${\bf u}{\bf\cdot}\nabla\times{\bf u}$ has this property). 
This is an important property of helicity as it measures parity 
symmetry breaking. A parity invariant flow has zero helicity. 
\section{Helicity in Fluids and MHD}

In this section we give a brief overview of helicity and vorticity 
conservation laws in ideal fluid dynamics and MHD. For 
ideal barotropic fluids, with no magnetic field we discuss the helicity 
conservation law involving the helicity density 
$h_f={\bf u}{\bf\cdot}\boldsymbol{\omega}$, where 
$\boldsymbol{\omega}=\nabla\times{\bf u}$ is the fluid vorticity.  
The integral $H_f=\int_{V_m} h_f \ d^3x$  
over a volume $V_m$ moving with the fluid, is known as the fluid 
helicity. It plays a key role in topological fluid dynamics in the description 
of the linkage of the vorticity streamlines (e.g. Moffatt (1969), 
Arnold and Khesin (1998)).  
The integral $H_m=\int_{V_m} {\bf A}{\bf\cdot} {\bf B}\ d^3x$ 
in MHD is known as the magnetic helicity, where ${\bf B}=\nabla\times{\bf A}$ 
is the magnetic induction and ${\bf A}$ is the magnetic vector potential. 
It describes the linkage of the 
magnetic field lines (Woltjer (1958), Berger and Field (1984)).  
A further quantity of interest in MHD is the cross 
helicity $H_c=\int_{V_m} {\bf u}{\bf\cdot}{\bf B}\ d^3x$ which describes 
the topology of the magnetic field and fluid velocity streamlines. 
One of the main aims of the present paper is to show how these fluid 
and MHD invariants are obtained by Lie dragging 
invariant differential forms and scalars with the flow 
(Tur and Janovsky 1993).  We also describe 
helicity and cross helicity conservation laws in MHD and gas dynamics.

\subsection{Helicity in Fluid Dynamics}

In a barotropic, ideal fluid, in which the pressure $p=p(\rho)$, is independent of the entropy $S$, the helicity density 
\beqn
h_f={\bf u}{\bf\cdot}\boldsymbol{\omega}\quad
\hbox{where}\quad \boldsymbol{\omega}=\nabla\times{\bf u}, \label{eq:hf1}
\eeqn
satisfies the helicity conservation law:
\beqn
\deriv{h_f}{t}+\nabla{\bf\cdot}\left[{\bf u}h_f
+\left(h-\frac{1}{2}|{\bf u}|^2\right) \boldsymbol{\omega}\right]=0. 
\label{eq:hf2}
\eeqn
The total helicity for a fluid volume $V_m$ moving with the 
fluid 
is conserved following the flow (e.g. Moffatt 1969). 
Thus, for a barotropic fluid,
\beqn
\frac{dH_f}{dt}=0\quad \hbox{where}\quad H_f=\int_{V_m} {\bf u}{\bf\cdot}
\nabla\times{\bf u}\ d^3x, \label{eq:hf3}
\eeqn
where $H_f$ is the total helicity of the fluid in the volume $V_m$. 
For the conservation law (\ref{eq:hf3}) to apply, it is required that 
that the component of the vorticity $\omega_n$ normal to the boundary 
$\partial V_m$ vanish on $\partial V_m$, i.e. 
$\omega_n=\boldsymbol{\omega}{\bf\cdot}{\bf n}=0$ on $\partial V_m$.
Here $d/dt=\partial/\partial t+{\bf u}{\bf\cdot}\nabla$ is 
the total Lagrangian 
time derivative following the flow. 

To derive (\ref{eq:hf2}), note that for a ideal 
gas, the momentum equation for the fluid:
\beqn
\deriv{\bf u}{t}+{\bf u}{\bf\cdot}\nabla{\bf u}=-\frac{1}{\rho}\nabla p, 
\label{eq:hf4}
\eeqn
can be written in the form:
\beqn
\deriv{\bf u}{t}-{\bf u}\times\boldsymbol{\omega}=T\nabla S
-\nabla\left(h+\frac{1}{2}|{\bf u}|^2\right). \label{eq:hf5}
\eeqn
For the case of a barotropic equation of state, there is no $T\nabla S$ 
term in (\ref{eq:hf5}). 

Taking the curl of the momentum equation (\ref{eq:hf5}) gives the 
vorticity equation:
\beqn
\deriv{\boldsymbol{\omega}}{t}-\nabla\times({\bf u}\times\boldsymbol{\omega})
=\nabla T\times\nabla S. \label{eq:hf9}
\eeqn
Taking the scalar product of $\boldsymbol{\omega}$ with the momentum equation 
(\ref{eq:hf5}) and adding the scalar product of ${\bf u}$ with the vorticity 
equation (\ref{eq:hf9})  gives the equation
\beqn
\deriv{({\bf u\cdot}\boldsymbol{\omega})}{t} 
+\nabla{\bf\cdot}
\left[({\bf u}{\bf\cdot}\boldsymbol{\omega}){\bf u}
+\left(h-\frac{1}{2}|{\bf u}|^2\right)\boldsymbol{\omega}\right]
=\boldsymbol{\omega}{\bf\cdot}(T\nabla S)
+{\bf u}{\bf\cdot}\nabla T\times\nabla S. \label{eq:hf10}
\eeqn
For a barotropic fluid there are no entropy gradients i.e. $\nabla S=0$, 
and in that case (\ref{eq:hf10}) reduces to the helicity conservation law 
(\ref{eq:hf2}). 
The integral conservation law (\ref{eq:hf3}) can be derived by using the 
identity: 
\beqn
\frac{d}{dt}\left(\frac{{\bf u}{\bf\cdot}\boldsymbol{\omega}}{\rho}\right)=-\frac{\boldsymbol{\omega}}{\rho}{\bf\cdot}\nabla\left(h-\frac{1}{2}|{\bf u}|^2\right). \label{eq:hf12}
\eeqn
The detailed proof of $dH_f/dt=0$ is given by Moffatt (1969). 

There are other conservation laws for ideal fluids. For 
example Kelvin's theorem implies that the circulation 
$\Gamma=\oint_{C} {\bf u}{\bf\cdot} d{\bf x}$ is conserved following the 
flow, for an ideal, barotropic fluid, where $C$ is a closed path moving 
with the fluid, i.e. $d\Gamma/dt=0$, where 
$d/dt=\partial/\partial t+{\bf u}{\bf\cdot}\nabla$ is the Lagrangian 
time derivative moving with the flow (this result also holds if there is 
a conservative, external gravitational field present).  However, 
the circulation is not 
conserved if there are entropy gradients in the flow, in 
which case $d\Gamma/dt=\int_{A}(\nabla T\times\nabla S){\bf\cdot}{\bf n}dA$, 
where 
$A$ is the area enclosing $C$ with normal ${\bf n}$. 

\begin{theorem}{\bf Ertel's Theorem}
Ertel's theorem in ideal fluid mechanics 
states that the potential vorticity 
$q=\boldsymbol{\omega}{\bf\cdot}\nabla S/\rho$  is a scalar invariant advected
with the flow, i.e., 
\beqn
\frac{d}{dt}\left(\frac{\boldsymbol{\omega}{\bf\cdot}\nabla S}{\rho}\right)=0, 
\label{eq:ertel1}
\eeqn
where $\boldsymbol{\omega}=\nabla\times{\bf u}$ is the fluid vorticity.
\end{theorem}

\begin{proof}
The vorticity equation (\ref{eq:hf9}) may be written in the form:
\beqn
\frac{d\boldsymbol{\omega}}{dt}+\boldsymbol{\omega}\nabla{\bf\cdot u}
-\boldsymbol{\omega}{\bf\cdot}\nabla{\bf u}=\nabla T\times\nabla S. 
\label{eq:ertel2}
\eeqn
Using the mass continuity equation (\ref{eq:2.1}) in the form 
$\nabla{\bf\cdot}{\bf u}=-(d\rho/dt)/\rho$ in (\ref{eq:ertel2}) 
gives the equation:
\beqn
\frac{d}{dt}\left(\frac{\boldsymbol{\omega}}{\rho}\right) 
-\frac{\boldsymbol{\omega}}{\rho}{\bf\cdot}\nabla{\bf u}
=\frac{\nabla T\times\nabla S}{\rho}. \label{eq:ertel3}
\eeqn
Also using the entropy advection equation $dS/dt=0$, we obtain:
\beqn
\frac{d}{dt}\nabla S=\nabla\left(\frac{dS}{dt}\right)
-(\nabla{\bf u})^T{\bf\cdot}\nabla S
\equiv -(\nabla{\bf u})^T{\bf\cdot}\nabla S.
\label{eq:ertel5}
\eeqn
Taking the scalar product of (\ref{eq:ertel3}) with $\nabla S$ and the 
scalar product of (\ref{eq:ertel5}) with $\boldsymbol{\omega}/\rho$ 
and adding the resultant equations gives Ertel's theorem (\ref{eq:ertel1}). 
This completes the proof.
\end{proof}

\subsection{Helicity in MHD}
Magnetic helicity in space and fusion plasmas has been investigated
 as a key quantity describing the topology
of magnetic fields
(e.g. Moffatt (1969,1978), Moffatt and Ricca (1992), Berger and Field (1984), 
Finn and Antonsen (1985,1988), Rosner et al. (1989), Low (2006),
Longcope and Malanushenko (2008)) 
 The magnetic helicity $H$ is defined as:
\beqn
H=\int_V\boldsymbol{\omega}_1\wedge d\boldsymbol{\omega}_1
=\int_V d^3x {\bf A\cdot B}, \label{eq:3.1}
\eeqn
where $\boldsymbol{\omega}_1={\bf A}{\bf\cdot}d{\bf x}$ 
is the magnetic vector potential one-form, $d\boldsymbol{\omega}_1={\bf B}{\bf\cdot} d{\bf S}$ is the magnetic field two-form;  
 ${\bf B}=\nabla\times{\bf A}$ is the magnetic induction,
 ${\bf A}$ is the magnetic vector potential and $V$ is the isolated volume
 in which the magnetic field configuration of interest is located. The
magnetic helicity is an invariant of magnetohydrodynamics (MHD)
(Els\"asser (1956), Woltjer (1958), Moffat (1969,1978)).
In (\ref{eq:3.1}) it is assumed that the normal magnetic 
field $B_n={\bf B\cdot n}$ vanishes on the boundary $\partial V$ of the 
volume $V$. The magnetic helicity (\ref{eq:3.1}) when expressed as an 
integral of  $\boldsymbol{\omega}_1\wedge d\boldsymbol{\omega}_1$  
is known as the Hopf invariant. 

For open-ended magnetic field configurations,
 a gauge independent definition of relative helicity
for a magnetic field configuration in a volume V (Finn and Antonsen (1985,1988))
is:
\beqn
H_r=\int_V d^3x \left({\bf A}_1+{\bf A}_2\right){\bf\cdot}
\left(\bf{B}_1-{\bf B}_2\right), \label{eq:3.2}
\eeqn
(see also Berger and Field (1984)
for an equivalent definition)
where ${\bf B}_1=\nabla\times{\bf A}_1$ describes the magnetic field of interest and ${\bf B}_2=\nabla\times {\bf A}_2$ is a reference magnetic field
with the same normal flux as ${\bf B}_1$ (in many applications the reference
magnetic field is a potential magnetic field, i.e. $\nabla\times{\bf B}_2=0$).
Relative  helicity is now commonly used  in the modeling of solar magnetic 
structures (Longcope and Malanushenko (2008), Low (2006)).  
Bieber et al. (1987) 
and Webb et al. (2010a) investigated the relative helicity of the Parker 
interplanetary spiral magnetic field. 
Berger and Ruzmaikin (2000) investigated the 
injection of magnetic helicity into the solar wind from the photospheric base
based on observational data and the differential rotation of the Sun. 
Webb et al. (2010b) obtained the relative helicity of shear and 
torsional Alfv\'en waves.

\subsubsection{Magnetic helicity conservation equation}
For ideal MHD,
$h_m={\bf A}{\bf\cdot}{\bf B}$
 satisfies the conservation law:
\begin{equation}
\deriv{h_m}{t}+\nabla{\bf\cdot}
\left[{\bf u}h_m+{\bf B}(\phi_E-{\bf A}{\bf\cdot}{\bf u})\right]=0, 
\label{eq:h2}
\end{equation}
where
\begin{equation}
{\bf E}=-\nabla\phi_E-\deriv{\bf A}{t}=-{\bf u}\times{\bf B},\quad 
{\bf B}=\nabla\times{\bf A}. \label{eq:h3}
\end{equation}
In (\ref{eq:h3}) ${\bf E}=-{\bf u}\times {\bf B}$ is Ohm's law 
in the infinite conductivity limit.

To derive (\ref{eq:h2}) use Faraday's law 
 ${\bf B}_t+\nabla\times{\bf E}=0$
and ${\bf B}=\nabla\times{\bf A}$ to obtain the equations: 
\beqnar
&&{\bf B}_t+\nabla\times {\bf E}=0,\label{eq:mh1a}\\
&&{\bf A}_t+{\bf E}+\nabla\phi_E=0, \label{eq:mh1b}
\eeqnar
where $\phi_E$ is the electric field potential. Note that 
the curl of (\ref{eq:mh1b}) gives Faraday's law (\ref{eq:mh1a}). 
 Combining (\ref{eq:mh1a})-(\ref{eq:mh1b})  
gives the equation:
\beqn
\derv{t}({\bf A}{\bf\cdot}{\bf B})+\nabla{\bf\cdot}
\left({\bf E}\times{\bf A}+\phi_E {\bf B}\right)
=-2{\bf E}{\bf\cdot}{\bf B}, \label{eq:mh4}
\eeqn
Noting  ${\bf E}=-{\bf u}\times{\bf B}$ and ${\bf E}{\bf\cdot}{\bf B}=0$
in (\ref{eq:mh4}) gives helicity
conservation equation (\ref{eq:h2}) for ideal MHD. 

The total magnetic helicity
$H_m=\int_{V_m} {\bf A}{\bf\cdot}{\bf B}\ d^3 x$ moving with the flow
is invariant, i.e. $dH_m/dt=0$ provided ${\bf B}{\bf\cdot}{\bf n}=0$
on the boundary surface $\partial V_m$ of the volume $V_m$ 
 The proof is similar to 
the proof that the helicity in ideal fluid mechanics:
$H_f=\int_{V_m}{\bf u}{\bf\cdot}\nabla\times {\bf u}\ d^3x$ is conserved
following the flow (see (\ref{eq:hf3}) and also Moffatt (1969,1978)).

Consider  the choice of the gauge for ${\bf A}$. 
By setting ${\bf B}=\nabla\times{\bf A}$,
 (\ref{eq:h3}) may be written in the form:
\beqn
\frac{d{\bf A}}{dt}=\nabla({\bf A}{\bf\cdot}{\bf u}-\phi_E)
-(\nabla{\bf u})^T{\bf\cdot}{\bf A}, \label{eq:2.27}
\eeqn
 where
$d/dt=\partial/\partial t+{\bf u}{\bf\cdot}\nabla$ is
the Lagrangian time derivative. Using  the gauge transformation:
\beqn
{\tilde{\bf A}}={\bf A}+\nabla\Lambda\quad\hbox{where}\quad
\frac{d\Lambda}{dt}+{\bf A}{\bf\cdot}{\bf u}-\phi_E=0, \label{eq:2.29}
\eeqn
in (\ref{eq:2.27}), 
Faraday's equation, for ${\tilde{\bf A}}$ reduces to:
\beqn
\frac{d{\tilde{\bf A}}}{dt}+(\nabla{\bf u})^T{\bf\cdot}{\tilde{\bf A}}=0,
\label{eq:2.30}
\eeqn
Equation (\ref{eq:2.30}) can also be written in the form:
\beqn
\deriv{\tilde{\bf A}}{t}
-{\bf u}\times(\nabla\times\tilde{\bf A})
+\nabla({\bf u}{\bf\cdot}\tilde{\bf A})=0. \label{eq:2.30a}
\eeqn
The latter equation is equivalent to (\ref{eq:h3}) for the electric 
field ${\bf E}=-{\bf u}\times{\bf B}$ in the form:
\beqn
{\bf E}=-\nabla({\bf u}{\bf\cdot}\tilde{\bf A})-\deriv{\tilde{\bf A}}{t}, 
\label{eq:2.30b}
\eeqn
which shows that the electric potential in the new gauge is $\tilde{\phi}_E
={\bf u}{\bf\cdot}\tilde{\bf A}$. The 
evolution equation (\ref{eq:2.30a}) is equivalent to the equation 
$d/dt(\tilde{\bf A}{\bf\cdot}d{\bf x})=0$ (see Section 4), which shows 
that the 1-form $\alpha=\tilde{\bf A}{\bf\cdot}d{\bf x}$ is Lie dragged 
by the flow. 
Combining (\ref{eq:2.30})
with Faraday's equation for ${\bf B}$ gives the helicity transport equation:
\beqn
\deriv{\tilde{h}}{t}+\nabla{\bf\cdot}({\tilde h} {\bf u})=0,
\label{eq:2.32}
\eeqn
where ${\tilde h}={\tilde{\bf A}}{\bf\cdot}{\bf B}$.

The gauge choice (\ref{eq:2.29}) appears to be the best choice of    
the gauge potential $\Lambda$ in the formulation of magnetic helicity
and magnetic helicity related conservation laws (Section 4), since it 
fits in with the idea that $\tilde{\bf A}{\bf\cdot}d{\bf x}$ is an 
invariant, Lie dragged one form, and gives the simplest continuity 
equation for the helicity conservation law (\ref{eq:2.32}). 

A question that naturally arises is what happens to 
 the advection equation 
(\ref{eq:2.30a}) in the limit as $|{\bf u}|\to 0$. 
Assuming $\partial \tilde{\bf A}/\partial t\sim \epsilon \hat{\bf n} g({\bf x})$
and $|{\bf u}|=\epsilon$, then in the limit as $\epsilon\to 0$, 
(\ref{eq:2.30a}) reduces to the equation:
\beqn
\epsilon\left[\hat{\bf n}g-\hat{\bf u}\times{\bf B}
+\nabla A_{\parallel}\right]=0, 
\quad\hbox{where}\quad A_\parallel={\bf A}{\bf\cdot}\hat{\bf u}, 
\label{eq:2.34}
\eeqn
where $\hat{\bf u}$ is a unit vector in the direction of ${\bf u}$. 
From (\ref{eq:2.34}) we obtain:
\beqnar
&&{\bf B}_\perp={\bf B}-({\bf B}{\bf\cdot}\hat{\bf u})\hat{\bf u}=
\left(\nabla A_{\parallel}+g({\bf x}){\bf n}\right)\times\hat{\bf u}, 
\nonumber\\ 
&&{\bf B}=B_\parallel\hat{\bf u}+{\bf B}_\perp,\quad {\bf B}_\perp
=\left(\nabla A_{\parallel}+g({\bf x}){\bf n}\right)\times\hat{\bf u}.
\label{eq:2.35}
\eeqnar
Equation (\ref{eq:2.35}) needs to be supplemented by Gauss's 
law $\nabla{\bf\cdot}{\bf B}=0$.  The split up of the field into parallel 
and perpendicular 
components is reminiscent of the poloidal-toroidal decomposition of 
the field. This decomposition of the field is 
useful in describing magnetic helicity, in which
 the helicity in given in terms of the linkage of the toroidal and 
poloidal magnetic fluxes, which is independent of the gauge for 
${\bf A}$ (e.g. Kruskal and Kulsrud (1958), Low (2006), Webb et al. (2010a,b)).
A simple example, where this decomposition occurs is the case of  
 magnetostatic equilibria 
in plane Cartesian geometry with an ignorable coordinate $z$, with 
$\hat{u}=(0,0,1)$ and with $g=0$, gives the magnetic field 
representation 
\beqn
{\bf B}_{\perp}=(B_x,B_y,0)
=\left(\deriv{A}{y},-\deriv{A}{x},0\right), \label{eq:2.36}
\eeqn
where $A=A_z$, which is consistent with solutions of 
the Grad-Shafranov equation 
for magnetostatic equilibria with an ignorable coordinate $z$. Our main 
concern here is what are the implications of taking $|{\bf u}|\to 0$
in the evolution equation (\ref{eq:2.30a}) for $\tilde{\bf A}$, 
and whether it is consistent to use this gauge in this limit? 
This issue needs to be investigated in further detail, but will not be 
pursued further in the present paper.

\subsubsection{Cross Helicity in MHD}
The cross helicity (for $p=p(\rho)$) is defined as the integral:
\begin{equation}
C[u,B]=\int_D d^3x\ {\bf u\cdot B}. \label{eq:crh1}
\end{equation}
It is a Casimir of barotropic MHD ($p=p(\rho)$),  
i.e. $\{F,C\}=0$, for any functional $F$ where $\{.,.\}$ 
is MHD Poisson brackets (Padhye and Morrison (1996)). 
 It is  also referred to as a rugged invariant of 
MHD (Matthaeus et al. (1982)). 
In order for the cross helicity 
to be an MHD invariant, it is implicitly assumed that 
${\bf B}{\bf\cdot n}=0$ on the boundary $\partial V$ of the volume $V$ 
of interest. It is straightforward  to adapt the argument (\ref{eq:hf3}) 
used to show the invariance or constancy of the fluid helicity $H_f$ 
to show that $dH_c/dt=0$ where $H_c=\int_{V_m} {\bf u\cdot B}\ d^3 x$
is the cross helicity for a volume $V_m$ moving with the fluid. 

The cross helicity density conservation law (for $p=p(\rho)$) is:
\begin{equation}
\deriv{h_c}{t}+\nabla{\bf\cdot}\left[{\bf u}h_c
+{\bf B}\left(h-\frac{1}{2}|{\bf u}|^2\right)\right]=0  \quad\hbox{where}\quad 
h_c={\bf u\cdot B}, \label{eq:crh2}
\end{equation}
and $h=(p+\varepsilon)/\rho$ is the gas enthalpy.
Equation (\ref{eq:crh2}) also holds
if $p=p(\rho,S)$ and ${\bf B}{\bf\cdot}\nabla S=0$.
Conservation 
law (\ref{eq:crh2}) is due to fluid relabelling symmetries. 
To derive  
(\ref{eq:crh2}) we use the Faraday and momentum equations:  
\beqn
\deriv{\bf B}{t}-\nabla\times({\bf u}\times{\bf B})=0,
\quad \frac{d{\bf u}}{dt}
=-\frac{1}{\rho}\nabla p+\frac{{\bf J}\times{\bf B}}{\rho}, \label{eq:crossh1} 
\eeqn
where ${\bf J}=\nabla\times{\bf B}/\mu_0$ and $d{\bf u}/dt=(\partial_t
+{\bf u\cdot}\nabla){\bf u}$ and we assume $\nabla{\bf\cdot}{\bf B}=0$.  
Using the first law of thermodynamics (\ref{eq:2.10}) 
 the momentum equation for the 
MHD fluid can be written in the form:
\beqn
{\bf u}_t-{\bf u}\times{\boldsymbol\omega}
+\nabla\left(h+\frac{1}{2}|{\bf u}|^2\right)
-\frac{{\bf J}\times{\bf B}}{\rho}- T\nabla S=0, \label{eq:crossh4}
\eeqn
where ${\bf u}_t=\partial {\bf u}/\partial t$. 
Taking the scalar product of Faraday's equation with ${\bf u}$ 
and the scalar product of the momentum equation (\ref{eq:crossh4})
with ${\bf B}$ gives the cross helicity equation:
\beqn
\derv{t}({\bf u\cdot B})+\nabla{\bf\cdot}\left[{\bf E}\times{\bf u}
+\left(h+(1/2)|{\bf u}|^2\right){\bf B}\right]=T{\bf B}{\bf\cdot}\nabla S.
\label{eq:crossh7}
\eeqn
If ${\bf B}{\bf\cdot}\nabla S=0$, equation (\ref{eq:crossh7}) reduces to the
cross-helicity conservation law (\ref{eq:crh2}).

The  helicity conservation equation 
 (\ref{eq:hf2}) holds for a barotropic gas, 
in which there are no entropy gradients. Similarly, the cross helicity 
conservation law (\ref{eq:crh2}) holds provided ${\bf B}{\bf\cdot}\nabla S=0$.


\section{Advected Invariants}
Tur and Janovsky (1993) developed a formalism for  
 geometrical objects ${\bf G}$ (tensors, p-forms and vectors) 
that are advected with the flow in ideal 
gas dynamics and MHD. The basic requirement for ${\bf G}$ 
to be advected or Lie dragged with the flow ${\bf u}$ is that 
\beqn
\left(\derv{t}+{\bf u\cdot}\nabla\right){\bf G}\equiv 
\left(\derv{t}+{\cal L}_{\bf u}\right) {\bf G}=0, 
\label{eq:ad1}
\eeqn
where ${\cal L}_{\bf u}$ denotes the Lie derivative with respect to the vector 
field ${\bf u}$. As in the Calculus of exterior differential forms 
and in differential geometry (e.g. Harrison and Estabrook (1971), 
Misner Thorne and Wheeler (1973), 
Fecko (2006)), vector fields ${\bf V}$ and one-forms $\alpha=A_i dx^i
\equiv {\bf A}{\bf\cdot}d{\bf x}$ are dual.  

\subsection{Exterior differential forms and vector fields}
Useful discussions of the algebra of exterior differential forms may be found 
for example in the books by Frankel (1997),Bott and Tu (1982),    
Misner at al. (1973), Marsden and Ratiu 
(1994), Holm (2008a,b), Flanders (1963). A useful short summary is 
given in the paper by Harrison and Estabrook (1971), who develop 
a geometric approach to invariance groups and solutions of 
partial differential systems using Cartan's geometric formulation of 
partial differential equations in the language 
of exterior differential forms and vector fields.  
This formalism was used by  Tur and Janovsky (1993) in their work 
on advected invariants in fluids and MHD plasmas.  
 
The vector field 
${\bf V}$ in 3D Cartesian geometry is thought of as a directional 
derivative operator:
\beqn
{\bf V}=V^x\derv{x}+V^y\derv{y}+V^z\derv{z}, \label{eq:ad2}
\eeqn
and the one-form ${\bf A}{\bf\cdot}d{\bf x}$ has the form:
\beqn
\boldsymbol{\omega}={\bf A}{\bf\cdot}d{\bf x}
=A_x dx+A_y dy+A_z dz. \label{eq:ad3}
\eeqn
The inner product of vector ${\bf u}$
and  1-form $\omega$ is the scalar or dot product:
\begin{align}
\langle{\bf u},\omega\rangle&=
\left\langle u^x\derv{x}+u^y\derv{y}+u^z\derv{z}, 
A_x dx+A_y dy+A_z dz\right\rangle\nonumber\\ 
&=u^x A_x+u^y A_y+u^z A_z
\equiv {\bf u}{\bf\cdot}{\bf A}. \label{eq:ad4}
\end{align}
Equivalent notations for the inner product are:
\begin{equation} 
\langle{\bf u},\omega\rangle
\equiv {\bf u}
\mathrel{\lrcorner}\omega
\equiv {\bf i}_{\bf u}\omega,
\label{eq:ad5}
\end{equation}
where ${\bf i}_{\bf u}$ denotes inner product of contravariant
field ${\bf u}$ with a covariant field $\omega$.
\begin{equation}
\left\langle\derv{x^i},d x^j\right\rangle=\delta_{ij}, 
\quad \hbox{e.g.}\quad 
\left\langle\derv{x},dx\right\rangle=1,
\quad \left\langle\derv{x},dy\right\rangle=0, \label{eq:ad6}
\end{equation}

In a differentiable manifold of dimension $n$, a $p$-form 
$\boldsymbol{\omega}$ can be thought of as a completely antisymmetric covariant $pth$ rank tensor, described by its anti-symmetric components 
$\omega_{\mu_1\ldots\mu_p}$. The $p$-form in general can be expressed 
in the form:
\beqn
\boldsymbol{\omega}=\omega_{\mu_1\ldots\mu_p}dx^{\mu_1}\wedge dx^{\mu_2}\wedge\ldots
\wedge dx^{\mu_p}, \label{eq:ad7}
\eeqn
One forms can be thought of as elements of the cotangent space at a point 
of a manifold, involving the mapping of the vector fields in 
the tangent space onto the reals. In (\ref{eq:ad7}) $\wedge$ denotes a 
non-commutative anti-symmetrized multiplication. 
At each point of the manifold, 
the forms have a Grassmann algebra defined by the properties of 
the wedge product operator $\wedge$. 
Over the manifold we may use the three operations of
of exterior differentiation, $d$, of contraction with a vector
field ${\bf V}$ (a contravariant vector $V^\mu$),
${\bf V}\contr\boldsymbol{\omega}$ , and of Lie derivative
with respect to ${\bf V}$, ${\cal L}_{\bf V}\boldsymbol{\omega}$. These
operations give forms of rank $p+1$, $p-1$ and $p$ respectively.


Misner et al. (1973) and Schutz (1980) and other texts emphasize the 
geometrical picture of the commutator of two vector fields 
in terms of the closure 
of the quadrilateral associated with the two vector fields. In Misner et al. 
(1973), the Faraday two form is visualized in terms of an egg-crate 
like structure. For example, the magnetic field two-form 
${\bf B}{\bf\cdot}d{\bf S}$, has a geometric structure associated with 
the oriented surface element $d{\bf S}$ and the vector field ${\bf B}$, 
which describes magnetic flux tubes. The lie dragging of vectors, differential 
forms and tensors involves the concept of parallel transport, in which the 
change in the geometric quantity of interest at a point along a Lie orbit or 
trajectory, must be pulled back to the initial point involved in the derivative
to make sense. It is also useful in some applications to use the dual of 
vector fields and forms using the hodge star formalism (e.g. Flanders (1963), 
Frankel (1997) , Fecko (2006)).

Some of the basic properties of the wedge product, $\wedge$, 
of the exterior derivative $d$ and the Lie derivative ${\cal L}_{\bf V}$ 
are given below. 

 Let $\omega$ be a $p$-form, $\sigma$ a $q$-form, $f$ a $0$-form,
$c$ a constant, ${\bf V}$ and ${\bf W}$ be vector fields, then:
\begin{align}
&\omega\wedge\sigma=(-1)^{pq}\sigma\wedge\omega,\nonumber\\ 
&d(\omega\wedge\sigma)=d\omega\wedge\sigma+(-1)^p\omega\wedge d\sigma, 
\nonumber\\ 
&dd\omega=0,\quad dc=0, \label{eq:ex14}\\
&\left({\bf V}+{\bf W}\right)\mathrel{\lrcorner}\omega
={\bf V}\mathrel{\lrcorner}\omega+{\bf W}\mathrel{\lrcorner}\omega, 
\quad (f{\bf V})\mathrel{\lrcorner}\omega
=f({\bf V}\mathrel{\lrcorner}\omega), \nonumber\\
&{\bf V}\mathrel{\lrcorner}(\omega\wedge\sigma)
=({\bf V}\mathrel{\lrcorner}\omega)\wedge\sigma
+(-1)^p\omega\wedge({\bf V}\mathrel{\lrcorner}\sigma). 
\label{eq:ex15}
\end{align}
Cartan's magic formula for the Lie derivative of the $p$ form $\omega$:
\begin{equation}
{\cal L}_{\bf V}\omega={\bf V}\mathrel{\lrcorner}d\omega
+d({\bf V}\mathrel{\lrcorner}\omega), \label{eq:ex16}
\end{equation}
is a particularly useful formula in applications. 
 Other Lie derivative formulae are:
\begin{align}
&{\cal L}_{\bf V}f={\bf V}\mathrel{\lrcorner}df, \quad 
{\cal L}_{\bf V}d\omega=d\left({\cal L}_{\bf V}\omega\right), \nonumber\\
&{\cal L}_{\bf V}\left(\omega\wedge\sigma\right)
=\left({\cal L}_{\bf V}\omega\right)\wedge\sigma
+\omega\wedge\left({\cal L}_{\bf V}\sigma\right), \nonumber\\
&{\cal L}_{\bf V}\left({\bf W}\mathrel{\lrcorner}\omega\right)=
[{\bf V},{\bf W}]\mathrel{\lrcorner}\omega
+{\bf W}\mathrel{\lrcorner}\left({\cal L}_{\bf V}\omega\right). 
\label{eq:ex17}
\end{align}

\leftline{\bf Exterior Derivative Formula Relations (vector notation)}

\beqnar
&df=\nabla f{\bf\cdot}d{\bf x},\nonumber\\
&d({\bf V}{\bf\cdot}d{\bf x})
=(\nabla\times{\bf V}){\bf\cdot}d{\bf S}\quad\hbox{(Stokes thm)}, \nonumber\\
&d({\bf A}{\bf\cdot}d{\bf S})=(\nabla{\bf\cdot}{\bf A}) dV
\quad\hbox{(Gauss thm)}, \nonumber\\
&d^2f=d(\nabla f{\bf\cdot}d{\bf x})
=(\nabla\times\nabla f){\bf\cdot}d{\bf S}=0\quad\hbox{(Poincar\'e Lemma)}, 
\nonumber\\
&d^2({\bf V}{\bf\cdot}d{\bf x})=d[(\nabla\times{\bf V}){\bf\cdot}d{\bf S}]
=\nabla{\bf\cdot}(\nabla\times {\bf V}) dV=0\quad \hbox{(Poincar\'e Lemma)}
\nonumber\\
&{\bf X}{\mathrel{\lrcorner}}({\bf V\cdot}d{\bf x})={\bf V\cdot X}, 
\nonumber\\
&{\bf X}{\mathrel{\lrcorner}}({\bf B\cdot}d{\bf S})
=-({\bf X}\times{\bf B}){\bf\cdot} d{\bf x},\nonumber\\
&{\bf X}{\mathrel{\lrcorner}}dV={\bf X\cdot}d{\bf S}, \nonumber\\
&d({\bf X}{\mathrel{\lrcorner}}dV)=d({\bf X\cdot}d{\bf S})
=(\nabla{\bf\cdot X}) dV. \label{eq:exterior1}
\eeqnar
$dd\omega=0$ for a $p$-form is known as the Poincar\'e Lemma.
It implies the equality of mixed second order partial derivatives.
If $\omega=d\alpha$ the form is exact and $d\omega=dd\alpha=0$.
A form with $d\omega=0$ is closed. Not all closed forms
are exact. Exactness means 'integrable'.

\leftline{\bf Lie Derivative Relations (vector notation)}
\beqnar
&{\cal L}_{\bf X}f={\bf X}{\mathrel{\lrcorner}}df
={\bf X}{\bf\cdot}\nabla f,\nonumber\\
&{\cal L}_{\bf X}({\bf V\cdot}d{\bf x})=\left(-{\bf X}
\times(\nabla\times{\bf V})+\nabla({\bf X}{\bf\cdot}V)\right)
{\bf\cdot}d{\bf x}, \nonumber\\
&{\cal L}_{\bf X}({\bf B\cdot}d{\bf S})=\left(-\nabla\times({\bf X}\times{\bf B})+{\bf X}(\nabla{\bf\cdot}{\bf B})\right){\bf\cdot}d{\bf S}, \nonumber\\
&{\cal L}_{\bf X}(fdV)=\nabla{\bf\cdot}({\bf X} f) dV, \label{eq:lie1}
\eeqnar
For vector fields ${\bf X}$ and ${\bf Y}$
\begin{equation}
{\cal L}_{\bf X}{\bf Y}=[{\bf X},{\bf Y}]
=({\bf X}{\bf\cdot}{\bf Y}-{\bf Y}{\bf\cdot}{\bf X}){\bf\cdot}
\nabla\equiv{\hbox{\rm ad}}_{\bf X}({\bf Y})\label{eq:lie2}
\end{equation}
Here $[{\bf X},{\bf Y}]$ is the left Lie bracket.

For a 1-form density $\sf{m}={\bf m\cdot}d{\bf x}\otimes dV$:
\begin{equation}
{\cal L}_{\bf X} \sf{m}=
\left(\nabla{\bf\cdot}({\bf X}\otimes {\bf m})
+(\nabla {\bf X})^T{\bf\cdot}{\bf m}\right){\bf\cdot}d{\bf x}\otimes dV
=:\sf{ad}_{\bf X}^*\sf{m}. \label{eq:lie3}
\end{equation}
The pairing between the one form density $\sf{m}$ 
and the vector field ${\bf u}$ is defined by the inner product:
\begin{equation}
\langle\sf{m},{\bf u}\rangle
=\int_\Omega {\bf u}{\mathrel{\lrcorner}}{\bf m}\ dV. \label{eq:lie4}
\end{equation}
Vector fields can be either left or right invariant vector fields. 
Thus, associated with the group transformation ${\bf x}=g{\bf x}_0$, 
the right invariant vector field 
${\bf u}=\dot{g}{\bf x}_0=\dot{g} g^{-1}{\bf x}$ defines the 
right invariant vector field ${\bf u}=\dot{g}g^{-1}$. The corresponding 
left invariant version of the same vector field is ${\bf v}=g^{-1}\dot{g}$.  
The right and left Lie brackets are related by:
 $[{\bf U},{\bf V}]_R=-[{\bf U},{\bf V}]_L$. The left
Lie bracket is used in (\ref{eq:lie2}). The right Lie bracket  used in
(\ref{eq:lie3}) is given by:
\begin{equation}
ad_{\bf U}({\bf V})=[{\bf U},{\bf V}]_R=\left({\bf V}{\bf\cdot}\nabla {\bf U}
-{\bf U\cdot}\nabla{\bf V}\right){\bf\cdot}\nabla. \label{eq:lie5}
\end{equation}
A more detailed discussion of the difference between right and left 
vector fields of a Lie algebra are given by Marsden and Ratiu (1994), 
Holm (1998), Holm (2008a,b) 
and Fecko (2006).  

\leftline{\bf Lie dragging of forms and vector fields}
Useful formulas for the Lie dragging of 0-forms, 1-forms, 2-forms, 3-forms 
and vector fields are given below. 
These formulae are particularly useful in describing advected 
invariants. 

\noindent{For 0-forms or functions $I$:}
\begin{equation}\frac{dI}{dt}
=\deriv{I}{t}+{\bf u}{\bf\cdot}\nabla I=0. \label{eq:liedr1}
\end{equation}
For 1-forms: ${\bf S}{\bf\cdot}d{\bf x}$
\begin{equation}
\frac{d}{dt}\left({\bf S}{\bf\cdot}d{\bf x}\right)=\left(\deriv{\bf S}{t}
-{\bf u}\times(\nabla\times {\bf S})
+\nabla({\bf u}{\bf\cdot}{\bf S})\right){\bf\cdot}d{\bf x}=0, \label{eq:liedr2}
\end{equation}
For 2-forms ${\bf B}{\bf\cdot}d{\bf S}$:
\begin{equation}
\frac{d}{dt}\left({\bf B}{\bf\cdot}d{\bf S}\right)=\left(
\deriv{\bf B}{t}-\nabla\times({\bf u}\times{\bf B})
+{\bf u}(\nabla{\bf\cdot}{\bf B})\right){\bf\cdot}d{\bf S}=0. 
\label{eq:liedr3}
\end{equation}
For 3-forms $\rho dx\wedge dy\wedge dz$:
\begin{equation}
\frac{d}{dt}\left(\rho dx\wedge dy\wedge dz\right)=
\left(\deriv{\rho}{t}+\nabla{\bf\cdot}(\rho {\bf u})\right) 
dx\wedge dy\wedge dz=0.
\label{eq:liedr4}
\end{equation}
 For vector fields (the dual of one-forms): ${\bf J}=J^i\nabla_i$:
\begin{equation}
\frac{d{\bf J}}{dt}=\deriv{\bf J}{t}+\left[{\bf u},{\bf J}\right]=0, 
\quad\hbox{where}\quad  
[{\bf u},{\bf J}]=
\left({\bf u}{\bf\cdot}\nabla J^i-{\bf J}{\bf\cdot}\nabla u^i\right)\nabla_i,
 \label{eq:liedr5} 
\end{equation}
is the Lie bracket of ${\bf u}$ and ${\bf J}$.
There are many invariants which are advected with the flow involving 
${\bf A}$, ${\bf B}$, $S$, and $\rho$ (e.g. Tur and Janovsky (1993).

\subsection{Applications}
 
Consider the quantities:
\begin{equation}
{\bf S}'=\nabla S({\bf x},t),\quad I'=\frac{{\bf A}{\bf\cdot B}}{\rho}, 
\quad \rho'={\bf A}{\bf\cdot}{\bf B}. \label{i1}
\end{equation}
One can show that 
\begin{equation}
\frac{d}{dt} I'=0, \quad 
\frac{d}{dt}\left({\bf S}'{\bf\cdot}d{\bf x}\right)=0, 
\quad \frac{d}{dt} ({\bf B}{\bf\cdot}d{\bf S})=0,
\quad \frac{d}{dt}\left(\rho'dx\wedge dy\wedge dz\right)=0,  \label{eq:i2}
\end{equation}
where 
\begin{equation}
\frac{d}{dt}=\derv{t}+{\bf u}{\bf\cdot}\nabla\equiv \derv{t}+{\cal L}_{\bf u}, 
\label{eq:i3}
\end{equation}
is the Lagrangian or advective time derivative  following the flow,  
and ${\cal L}_{\bf u}$ denotes the Lie derivative with respect to the vector 
field ${\bf u}$.  Here $I'$ is a scalar or $0-$form,  
$\nabla S{\bf\cdot}d{\bf x}$ is a 1-form, ${\bf B}{\bf\cdot}d{\bf S}$ 
is a 2-form and $\rho'dx\wedge dy\wedge dz$ is a 3-form, which are advected 
invariants that are Lie dragged by the flow (i.e. these quantities 
remain invariant moving with the flow).  
The advection invariance of the Faraday 2-form ${\bf B}{\bf\cdot} d{\bf S}$ 
is equivalent to Faraday's equation (\ref{eq:2.4}).  
There are many other invariants. Some integral invariants are:
\beqnar
&&\Gamma^1_1=\oint_{\gamma(t)}\Phi {\bf A\cdot}d{\bf l},\quad
\Gamma^2_1=\int_{\bf S(t)} \Phi {\bf B}{\bf\cdot} d{\bf S}', \quad
I^3_2=\int_{\Omega(t)} \Phi ({\bf A\cdot B})\ d^3 x,\nonumber\\
&&I_3^4=\int_{\Omega(t)}\Phi {\bf A}{\bf\cdot}
\left[\nabla S\times
\nabla\left(\frac{{\bf A}{\bf\cdot}{\bf B}}{\rho}\right)\right]\ d^3x,
\nonumber\\
&&\Phi=\Phi\biggl(\frac{{\bf A}{\bf\cdot}{\bf B}}{\rho}, S,
\frac{\bf B}{\bf A\cdot B}{\bf\cdot}\nabla
\left(\frac{\bf A\cdot B}{\rho}\right),
\frac{\bf B}{\rho}{\bf\cdot}\nabla
\left(\frac{{\bf B\cdot}\nabla S}{\rho}\right)\ldots\biggr), \label{eq:i4}
\eeqnar
where $\Phi$ is an arbitrary function of its arguments.

\subsection{Lie Dragging}
\begin{example}{\bf 1.}
Consider the results of Lie dragging the 
Faraday 2-form:
\begin{equation}
\beta=B_x dy\wedge dz+B_y dz\wedge dx
+B_zdx\wedge dy\equiv {\bf B}{\bf\cdot} d{\bf S}, \label{eq:22}
\end{equation}
where ${\bf u}=u^x\partial_x+u^y\partial_y+u^z\partial_z
={\bf u}{\bf\cdot}\nabla$ is the vector field representing the fluid 
velocity (here we use the notation of vector fields and one-forms 
used in modern differential geometry (e.g. Misner et al. (1973)), 
in which the base vectors for contravariant vector fields
are written as $\partial_{x^i}\equiv {\bf e}_i$ and the base vectors 
for one-forms or covariant vector fields are written as $dx^i\equiv {\bf e}^i$) 
We use Cartan's magic formula:
\begin{equation}
{\cal L}_{\bf u}(\beta)={\bf u}\mathrel{\lrcorner} d\beta
+d({\bf u}\mathrel{\lrcorner}\beta). \label{eq:23}
\end{equation}
Calculating  $d\beta$ and ${\bf u}\mathrel{\lrcorner} d\beta$,
${\bf u}\mathrel{\lrcorner}\beta$ and $d({\bf u}\mathrel{\lrcorner}\beta)$
gives:
\beqnar
&&d\beta=
\nabla{\bf\cdot}{\bf B}\ dx\wedge dy\wedge dz, \quad
{\bf u}\mathrel{\lrcorner}d\beta
=\nabla{\bf\cdot}{\bf B} \left({\bf u}{\bf\cdot}d{\bf S}\right), \nonumber\\
&&{\bf u}\mathrel{\lrcorner}\beta=-({\bf u}\times {\bf B}){\bf\cdot}d{\bf x},
\quad
d({\bf u}\mathrel{\lrcorner}\beta)=-\nabla\times({\bf u}\times{\bf B})
{\bf\cdot}d{\bf S}. \label{eq:24}
\eeqnar
Using the results (\ref{eq:24}) in Cartan's formula (\ref{eq:23}) gives
\begin{equation}
\left(\derv{t}+{\cal L}_{\bf u}\right) \beta=\left(\deriv{\bf B}{t}
-\nabla\times({\bf u}\times{\bf B})+{\bf u}\nabla{\bf\cdot}{\bf B}\right)
{\bf\cdot} d{\bf S}=0, \label{eq:25}
\end{equation}
which implies Faraday's equation (note $\nabla{\bf\cdot}{\bf B}$
is advected with the
flow if $\nabla{\bf\cdot}{\bf B}\neq 0$ (e.g. in numerical MHD)).
\end{example}

\begin{example}{\bf 2.}
Consider the effect of Lie-dragging the 1-form:
\begin{equation}
\alpha=A_x dx+A_y dy+A_z dz\equiv {\bf A}{\bf\cdot}d{\bf x}. \label{eq:26}
\end{equation}
Using Cartan's magic formula:
\beqn
{\cal L}_{\bf u}(\alpha)={\bf u}\mathrel{\lrcorner}d\alpha
+ d({\bf u}\mathrel{\lrcorner}\alpha),  \label{eq:27}
\eeqn
and the results
\begin{align}
d\alpha&=(\nabla\times{\bf A}){\bf\cdot}d{\bf S}, 
\quad {\bf u}\mathrel{\lrcorner}d\alpha
=-[{\bf u}\times(\nabla\times{\bf A})]{\bf\cdot}d{\bf x}, \nonumber\\
{\bf u}\mathrel{\lrcorner}\alpha&=({\bf u}{\bf\cdot}{\bf A}),\quad 
d({\bf u}\mathrel{\lrcorner}\alpha)
=\nabla({\bf u}{\bf\cdot}{\bf A})
{\bf\cdot}d{\bf x}, \label{eq:28}
\end{align}
we obtain
\begin{equation}
\left(\derv{t}+{\cal L}_{\bf u}\right)\alpha=
\left(\deriv{\bf A}{t}-{\bf u}\times(\nabla\times{\bf A})
+\nabla({\bf u}{\bf\cdot}{\bf A})\right){\bf\cdot}d{\bf x}. \label{eq:29}
\end{equation}
If $(\partial_t+{\cal L}_{\bf u})\alpha =0$, then $\alpha$ is Lie dragged 
with the flow. Comparing (\ref{eq:29}) with (\ref{eq:2.30a})-(\ref{eq:2.30b}) 
it follows that ${\bf A}{\bf\cdot}d{\bf x}$ is Lie dragged by the flow 
if $\phi_E={\bf u}{\bf\cdot}{\bf A}$. In this special gauge 
${\bf A}{\bf\cdot}{\bf B}/\rho$ is an advected invariant
 (see (\ref{eq:maghel2})).  
\end{example}

\begin{example}{\bf 3.}
Faraday's equation (\ref{eq:2.4}) combined with the mass 
continuity equation (\ref{eq:2.1}) implies:
\begin{equation}
\deriv{\bf b}{t}+[{\bf u},{\bf b}]
\equiv \left(\derv{t}+{\cal L}_{\bf u}\right){\bf b}=0  \quad\hbox{where}\quad 
{\bf b}=\frac{\bf B}{\rho},
 \label{eq:31}
\end{equation}
and $[{\bf u},{\bf b}]$ is the Lie bracket of the vector fields 
${\bf u}$ and ${\bf b}$, i.e.
\begin{equation}
{\cal L}_{\bf u}({\bf b})=[{\bf u},{\bf b}]=[{\bf u},{\bf b}]^i\nabla_i 
 =\left(u^s\nabla_sb^i-b^s\nabla_su^i\right)\nabla_i. \label{eq:31a}
\end{equation} 
The vector field ${\bf b}$ is Lie dragged
with the fluid, and hence
\begin{equation}
b^i\derv{x^i}=b^j_0\derv{x^j_0}\quad \hbox{or}\quad  b^i=
x_{ij}b_0^j, \quad x_{ij}=\deriv{x^i}{x_0^j}, 
 \label{eq:32}
\end{equation}
where ${\bf x}={\bf x}({\bf x}_0,t)$ is the Lagrangian map.
 From (\ref{eq:31}) and (\ref{eq:32})
we obtain:
\begin{equation}
B^i=\frac{x_{ij} B_0^j({\bf x}_0)}{J}\quad \hbox{where}\quad J=\det(x_{ij}), 
\quad \label{eq:33}
\end{equation}
which is Cauchy's solution for ${\bf B}$ (e.g. Newcomb (1962),
Parker (1979)).
\end{example} 
\subsubsection{Entropy and mass advection}

The  entropy $S=S(x_0)$, is a 0-form (i.e. a function) which is Lie dragged 
with the fluid, i.e. 
\begin{equation}
\left(\derv{t}+{\cal L}_{\bf u}\right)S\equiv 
\deriv{S}{t}+{\bf u}{\bf\cdot}\nabla S=0, \label{eq:f15}
\end{equation}
which is  (\ref{eq:liedr1}) for the advection of a 0-form $I$, but with 
$ I\to S$. The integral of (\ref{eq:f15}) is $S=S_0({\bf x}_0)$, 
where ${\bf x}_0$ is the Lagrange fluid label for which ${\bf x}={\bf x}_0$
at time $t=0$. 

Consider the mass 3-form:
\begin{equation}
\beta=\rho\ dx\wedge dy\wedge dz. \label{eq:f16}
\end{equation}
  Using Cartan's formula (\ref{eq:23}) 
we find
$d\beta=0$ as $\beta$ is a 3-form in 3D xyz-space, and 
${\bf u}\mathrel{\lrcorner}\beta=\rho {\bf u}{\bf\cdot}d{\bf S}$, 
which implies:
\beqn
{\cal L}_{\bf u}(\beta)=0+d({\bf u}\mathrel{\lrcorner}\beta)
=\nabla{\bf\cdot}(\rho{\bf u})dx\wedge dy\wedge dz\label{eq:f17}
\eeqn
and
\beqn
\left(\derv{t}+{\cal L}_{\bf u}\right)\beta=
\left(\deriv{\rho}{t}+
\nabla{\bf\cdot}(\rho {\bf u})\right)d^3x=0.\label{eq:f18}
\eeqn
Equation (\ref{eq:f18}) is the same as (\ref{eq:liedr4}) for an advected 
3-form $\rho dx\wedge dy\wedge dz$. 
The integral of (\ref{eq:f18}) is:
\beqn 
\rho d^3x=\rho_0d^3x_0,\quad\hbox{where}\quad \rho=\rho_0({\bf x}_0)/J, 
\quad J=\det(x_{ij}). 
 \label{eq:f19}
\eeqn

Thus the mass continuity, entropy advection and Faraday's equation can all be 
expressed in terms of the Lie dragging of forms by the vector field ${\bf u}$.

\subsection{Theorems for Advected Invariants}
\begin{theorem}\label{thm1}
If $\omega^p$ is an invariant, then $\omega^{p+1}=d\omega^p$ 
is an invariant $(p+1)$-form.
\end{theorem}
\begin{proof}
{\bf Proof}\\
$\omega^p$ is invariant implies:
\begin{equation}
\left(\derv{t}+{\cal L}_{\bf u}\right)\omega^p=0. \label{eq:f20}
\end{equation}
Take $d$ of (\ref{eq:f20}).
Use $d\partial_t=\partial_t d$, and $d{\cal L}_{\bf u}={\cal L}_{\bf u} d$
gives (\ref{eq:f20}) but with $\omega^p\to \omega^{p+1}$.
\end{proof}

\begin{example}
The entropy $S$ is a scalar invariant implies
$\alpha=dS=\nabla S{\bf\cdot}d{\bf x}$ is a conserved advected 1-form. 
\end{example}
\begin{theorem}\label{thm2}
Let $\omega_1^k$ and $\omega_2^l$ be
advected $k$ and $l$-form invariants,
then $\omega^{k+l}=\omega_1^k\wedge\omega_2^l$ is an advected $(k+l)$-form
invariant.
\end{theorem}

\begin{proof}
{\bf Proof}:\\
Use
\beqnar
&&\derv{t}\left(\omega_1\wedge\omega_2\right)
=\deriv{\omega_1}{t}\wedge\omega_2+\omega_1\wedge\deriv{\omega_2}{t}\nonumber\\
&&{\cal L}_{\bf u}\left(\omega_1\wedge\omega_2\right)
={\cal L}_{\bf u}\left(\omega_1\right)\wedge\omega_2
+\omega_1\wedge{\cal L}_{\bf u}\left(\omega_2\right), \label{eq:f23}
\eeqnar
to get
\begin{equation}
\left(\derv{t}+{\cal L}_{\bf u}\right)(\omega_1\wedge\omega_2)=0. \label{eq:f24}
\end{equation}
\end{proof}
\begin{example}
$\omega_1={\bf S}_1{\bf\cdot}d{\bf x}$ and
$\omega_2={\bf S}_2{\bf\cdot}d{\bf x}$ are advected one-forms, then
\begin{equation}
\omega_1\wedge \omega_2=({\bf S}_1\times{\bf S}_2){\bf\cdot}d{\bf S}, 
\label{eq:f25}
\end{equation}
is an advected 2-form, and $(\omega_1\wedge\omega_2)/\rho$ is an advected
invariant vector field.
\end{example}

There are further theorems on the formation of advected invariants from known 
advected invariants described by Tur and Janovsky (1993). Some of these 
theorems are listed below, without proof. Cartan's magic formula is 
useful in many of the proofs.
\begin{theorem}\label{thm3}
If $\omega$ is a conserved $p$-form, and ${\bf J}$ is a conserved
vector,  then
$\omega^{(p-1)}={\bf J}\mathrel{\lrcorner}\omega$ is
a conserved $(p-1)$ form.
\end{theorem}
\begin{theorem}\label{thm4}
If $\omega$ is an invariant  $p$-form, and ${\bf J}$ is an invariant
vector field, then $\omega'={\cal L}_{\bf J}\omega$
is an invariant $p$-form.
\end{theorem}

\begin{theorem}\label{thm5}
If ${\bf J}_1$ and ${\bf J}_2$ are invariant vector fields then so
is $[{\bf J}_1,{\bf J}_2]$ iff $\{{\bf J}_1,{\bf J}_2, {\bf u}\}$ are
elements of a Lie algebra. 
\end{theorem}
\leftline{\bf Comment:} The question of Lie algebraic structures for 
fluid relabeling symmetries has been addressed by Volkov, Tur and Janovsky 
(1995). Their work shows that there is a hidden 
supersymmetry in hydrodynamical systems (i.e. ideal MHD and hydrodynamics), 
with respect to the odd Buttin bracket. 
 
\subsection{Magnetic Helicity}
 $\alpha={\bf A}{\bf\cdot}d{\bf x}$ is advected one 
form for the magnetic vector potential, provided the gauge for ${\bf A}$ 
is chosen so that $\phi_E={\bf u}{\bf\cdot}{\bf A}$. The Lie dragging 
condition for $\alpha={\bf A}{\bf\cdot}d{\bf x}$ implies:
\begin{equation}
\deriv{\bf A}{t}-{\bf u}\times(\nabla\times{\bf A})
+\nabla({\bf u}{\bf\cdot}{\bf A})=0. \label{eq:maghel1}
\end{equation}
This equation can be written as 
$d{\bf A}/dt+(\nabla{\bf u})^T{\bf\cdot}{\bf A}=0$. 
The  magnetic flux 2-form $\beta={\bf B}{\bf\cdot}d{\bf S}$  
and  the vector field 
${\bf b}={\bf B}/\rho$ are Lie dragged with the flow. Thus,   
${\bf b}\mathrel{\lrcorner}({\bf A}{\bf\cdot}d{\bf x})\equiv {\bf A\cdot B}
/{\rho}$ is a Lie dragged scalar invariant. 
Thus, we obtain the magnetic helicity conservation law:
\begin{equation}
\frac{d}{dt}\left(\frac{{\bf A}{\bf\cdot}{\bf B}}{\rho}\right)
=0\quad \hbox{or}\quad \deriv{h_m}{t}+\nabla{\bf\cdot}(h_m{\bf u})=0, 
\label{eq:maghel2}
\end{equation}
where $h_m={\bf A}{\bf\cdot}{\bf B}$  is the magnetic helicity in 
the gauge $\phi_E={\bf u}{\bf\cdot}{\bf A}$. 

\subsection{The Ertel invariant and related invariants}
In this section we discuss Ertel's theorem
in gas dynamics, and the generalization of Ertel's equation to MHD
(e.g. Kats 2003). The MHD generalization of Ertel's theorem uses
the Clebsch variable representation of the fluid velocity, that arises
from using Lagrangian constraints in the variational principle
for MHD discussed by Zakharov and Kuznetsov (1997).
We also discuss the Hollmann (1964) invariant, which is related to the
Ertel invariant (e.g. Tur and Yanovsky (1993)).
 The Ertel invariant is:
\begin{equation}
I_e=\frac{\omega{\bf\cdot}\nabla S}{\rho}\quad 
\hbox{where}\quad  \omega=\nabla\times{\bf u}. \label{eq:ert1}
\end{equation}
To derive the Ertel invariant we use the Clebsch representation for ${\bf u}$:
\beqnar
&&{\bf u}=\nabla\phi-r\nabla S-\lambda\nabla\mu,\nonumber\\
&&\phi=\int_0^t \left(\frac{1}{2}|{\bf u}|^2-h\right)({\bf x}_0,t')\ dt',
\quad r=-\int_0^t T_0({\bf x}_0,t')\ dt',\label{eq:ert2}
\eeqnar
where $h=(p+\varepsilon)/\rho=$ is the enthalpy, $S$ is the entropy,
$\phi$ is the velocity potential, and $T_0({\bf x}_0,t)=T({\bf x},t)$
is the temperature.
   $\lambda$ and $\mu$ are related to the
Lin constraints associated with vorticity
in a Lagrangian variational principle
with constraints (e.g. Zakharov and Kuznetsov (1997)).
The Clebsch variable representation for ${\bf u}$ 
is related to Weber transformations.

 Let
\begin{equation}
{\bf w}={\bf u}-\nabla\phi+r\nabla S\equiv -\lambda\nabla\mu, \label{eq:ert3}
\end{equation}
 $\nabla\times{\bf w}=-\nabla\lambda\times\nabla\mu$ represents
the component of the vorticity of the fluid that is not generated by
entropy gradients, i.e. it does not depend on $\nabla S$.
The one-form $\alpha={\bf w}{\bf\cdot}d{\bf x}$ is Lie dragged with the fluid.
Thus  ${\bf w}$ satisfies the equation (\ref{eq:liedr2}): 
\begin{equation}
\deriv{\bf w}{t}-{\bf u}\times(\nabla\times{\bf w})
+\nabla({\bf u\cdot w})=0. \label{eq:ert4}
\end{equation}

It follows that   
${\bf b}=(\nabla\times{\bf w})/\rho$ is a Lie dragged vector field and
$\nabla S{\bf\cdot}d{\bf x}$ is a
conserved 1-form (Tur and Janovsky (1993)).
Thus, 
${\bf b}{\mathrel{\lrcorner}}(\nabla S{\bf\cdot}d{\bf x})
={\bf b}{\bf\cdot}\nabla S$ is a conserved scalar. Inspection of 
${\bf b}{\bf\cdot}\nabla S$ reveals that:
\begin{equation}
I_e\equiv {\bf b}{\bf\cdot}\nabla S
=\frac{\nabla\times({\bf u}+r\nabla S-\nabla\phi)}{\rho}{\bf\cdot}\nabla S
=\frac{\nabla\times{\bf u}}{\rho}{\bf\cdot}\nabla S, \label{eq:ert5}
\end{equation}
 is the Ertel invariant.

\begin{theorem}\label{thm6}
 The generalization for the Ertel invariant in MHD is (Kats (2003)):
\begin{equation}
I_e^{(m)}=\frac{\nabla\times({\bf u}-{\bf u}_M)}{\rho}{\bf\cdot}\nabla S, 
\label{eq:ert6}
\end{equation}
where
\beqnar
&&{\bf u}_M=-\frac{(\nabla\times{\boldsymbol{\Gamma}})\times{\bf B}}{\rho}
-{\boldsymbol{\Gamma}}\frac{(\nabla{\bf\cdot}{\bf B})}{\rho},\label{eq:ert6a}\\
&&\deriv{\boldsymbol{\Gamma}}{t}
-{\bf u}\times(\nabla\times{\boldsymbol{\Gamma}})
+\nabla({\boldsymbol{\Gamma}}{\bf\cdot}{\bf u})=-\frac{\bf B}{\mu_0}, 
\label{eq:ert7}
\eeqnar
and $\mu_0$ is the magnetic permeability. We can also write (\ref{eq:ert7})
as:
\beqn
\frac{d}{dt}\left({\boldsymbol{\Gamma}}{\bf\cdot}d{\bf x}\right)
=-\frac{{\bf B}{\bf\cdot}d{\bf x}}{\mu_0}. \label{eq:ert7a}
\eeqn
\end{theorem}
\begin{proof}
 Use the Clebsch representation for ${\bf u}$:
\begin{equation}
{\bf u}=\nabla\phi-r\nabla S-\tilde{\lambda}\nabla\mu +{\bf u}_M.
 \label{eq:ert8}
\end{equation}
Inspection shows that ${\bf w}$ satisfies the equation:
\begin{equation}
{\bf w}={\bf u}-\nabla\phi+r\nabla S-{\bf u}_M\equiv -\tilde{\lambda}\nabla\mu,
\label{eq:ert9}
\end{equation}
and hence $\alpha={\bf w}{\bf\cdot}d{\bf x}$ is an invariant 1-form. 
It follows that  ${\bf b}=\nabla\times{\bf w}/\rho$ is a Lie advected
vector field. 
$dS=\nabla S{\bf\cdot}d{\bf x}$ is an invariant advected 1-form.
Thus, $I_e^m={\bf b}{\bf\cdot}\nabla S$ is an invariant scalar 
given by:
\beqnar
I_e^m&&=\nabla\times
\left({\bf u}-\nabla\phi+r\nabla S-{\bf u}_M\right){\bf\cdot}\nabla S/\rho
\nonumber\\
&&=\left[\nabla\times({\bf u}-{\bf u}_M)+\nabla r\times\nabla S\right]
{\bf\cdot}\nabla S/\rho\nonumber\\
&&\equiv \nabla\times({\bf u}-{\bf u}_M)
{\bf\cdot}\nabla S/\rho. \nonumber\\
&&\label{eq:ert10}
\eeqnar
The quantity $I_e^m$ is the MHD analogue of the Ertel invariant. It reduces 
to the Ertel invariant in the case where ${\bf u}_{M}$ is 
zero. 
\end{proof}

\begin{theorem}\label{thm7}
 The Hollmann invariant is:
\begin{equation}
I_h=({\bf u}-\nabla\phi){\bf\cdot}\frac{\nabla S\times\nabla I_e}{\rho}\quad 
\hbox{where}\quad I_e=\frac{(\nabla\times{\bf u}){\bf \cdot}\nabla S}{\rho}, 
\label{eq:hol1}
\end{equation}
is the Ertel invariant. Here $\phi$ is the Clebsch potential  
in (\ref{eq:ert2}) associated with potential flow. The Hollmann invariant $I_h$ 
is Lie dragged with the flow. 
\end{theorem}

\begin{proof}
 $\omega_1=\nabla S{\bf\cdot}d{\bf x}$ and
$\omega_2=\nabla I_e{\bf\cdot}d{\bf x}$ are conserved one-forms. Thus,
$\omega=\omega_1\wedge \omega_2=(\nabla S\times\nabla I_e){\bf\cdot}d{\bf S}$
is a conserved two form, and
\begin{equation}
{\bf b}=\nabla S\times\nabla I_e/\rho, \label{eq:hol2}
\end{equation}
is a conserved vector.
$\alpha={\bf w}{\bf\cdot}d{\bf x}$ is a conserved one-form, where
\begin{equation}
{\bf w}={\bf u}-\nabla\phi+r\nabla S, \label{eq:hol3}
\end{equation}
and ${\bf w}$ satisfies the equation:
\begin{equation}
\deriv{\bf w}{t}-{\bf u}\times(\nabla\times{\bf w})
+\nabla({\bf u}{\bf\cdot}{\bf w})=0. \label{eq:hol4}
\end{equation}

 Using (\ref{eq:hol2}) and (\ref{eq:hol3}) it follows that
\begin{equation}
I_h={\bf w}{\bf\cdot}{\bf b}\equiv ({\bf u}-\nabla\phi){\bf\cdot}
\frac{\nabla S\times\nabla I_e}{\rho}, \label{eq:hol5}
\end{equation}
is a scalar invariant (i.e. the Hollmann invariant).
\quad
\end{proof}
Similarly, 
the MHD version of the Hollmann invariant is:
\beqn
I_h^m={\bf w}_m{\bf\cdot}{\bf b}_m =\left({\bf u}-{\bf u}_M-\nabla\phi\right)
{\bf\cdot}\left(\frac{\nabla S\times\nabla I_e^m}{\rho}\right), \label{eq:hol6}
\eeqn
where
\beqn
{\bf w}_m={\bf u}-\nabla\phi+r\nabla S-{\bf u}_M, \quad 
{\bf b}_m=\frac{\nabla S\times\nabla I_e^m}{\rho}. \label{eq:hol7}
\eeqn
\subsection{Topological Invariants}

Topological invariants and integrals of differential forms over a volume $V$
that are non-zero are sometimes referred to as topological charges.
A more complete discussion is given by Tur and Yanovsky (1993).  First we recall the definitions of closed and exact differential forms.

\begin{definition}
A $p-$form $\omega^p$ is closed if its exterior derivative $d\omega^p=0$.
\end{definition}

\begin{definition}
A $p-$form $\omega^p$ is exact if it can be expressed as the exterior
derivative of a $(p-1)$-from $\omega^{p-1}$,
i.e., $\omega^p=d\omega^{p-1}$.
It is assumed that $\omega^p$ and $\omega^{p-1}$
are sufficiently smooth and differentiable
on a star-shaped region of the manifold on which the forms are defined.
\end{definition}

\begin{lemma}[Poincar\'e]
The Poincar\'e Lemma states that if $X$ is a contractible open set of
$R^n$, then any closed $p-$form defined on $X$ is exact, for any integer
$0<p\leq n$.
\end{lemma}

\begin{definition}
Contractibility means that there is a homotopy $F_t:X\times[0,1]\to X$
 that continuously deforms $X$ to a point. Thus every cycle $c$
in $X$ is the boundary of some cone. One can take the cone to be the image
of $c$ under the homotopy. A dual version of this result gives
the Poincar\'e Lemma.
\end{definition}
From the above definitions, it follows that an exact $p$-form is closed, but a
closed $p$-form is not necessarily exact. To verify these statements,
note that if $\omega^p$ is exact, then $\omega^p=d\omega^{p-1}$ for some
$p-1$ form $\omega^{p-1}$. By the Poincar\'e
Lemma, $d\omega^p=dd\omega^{p-1}=0$(i.e. the Poincar\'e Lemma
states that $dd\alpha=0$ for a differential
form $\alpha$, where $\alpha$ is sufficiently differentiable, i.e.
at least twice differentiable on the star shaped region of the manifold $M$ on
which the form is defined).
However, a closed form $\omega^p$ with $d\omega^p=0$ is not necessarily exact,
i.e. there might not exist a $(p-1)$ form such that $\omega^p=d\omega^{p-1}$.
The word exact is synonymous with the notion of global integrability.

An invariant integral of the form (e.g. magnetic helicity):
\beqn
I=\int_V \boldsymbol{\omega}\wedge d\boldsymbol{\omega}, \label{eq:top0}
\eeqn
where $\omega=\tilde{\bf A}{\bf\cdot}d{\bf x}$ is an advected invariant 1-form, and with $d\boldsymbol{\omega}
=\nabla\times\tilde{\bf A}{\bf\cdot}d{\bf S}$ for which 
the integral (\ref{eq:top0}) is non-zero, defines 
a non-zero topological charge known as the Hopf invariant. A classical 
example of an MHD solution with non-zero topological charge 
is the MHD topological 
soliton (e.g. Kamchatnov (1982)) and related topological MHD 
solutions (Semenov et al. (2002)). 

If $\beta=\omega{\bf\cdot}d{\bf S}$ is an advected invariant 2-form, then
${\bf J}=\boldsymbol{\omega}/\rho\equiv (\omega^i/\rho)\partial/\partial x^i$
is an invariant advected vector field, and
$d\beta=\nabla{\bf\cdot}\boldsymbol{\omega}\ d^3x\equiv
\nabla{\bf\cdot}(\rho {\bf J}) d^3x\neq 0$
if $\nabla{\bf\cdot}(\rho{\bf J})\neq 0$.
If $\nabla{\bf\cdot}(\rho{\bf J})\neq 0$,
the integral $I^q=\int\ d\beta$  has non-zero {\em topological charge}.
Examples of two-forms with non-zero topological charge can be constructed
from the wedge product of two invariant 1-forms. For example, if
\beqn
\omega_{S_1}^1={\bf S}_1{\bf\cdot} d{\bf x}, \quad
\omega_{S_2}^1={\bf S}_2{\bf\cdot} d{\bf x}, \label{eq:top1}
\eeqn
are invariant 1-forms, then
\beqn
\omega^2=\omega_{S_1}^1\wedge \omega_{S_2}^1
={\bf S}_1{\bf\cdot}d{\bf x}\wedge {\bf S}_2{\bf\cdot}d{\bf x}
=({\bf S}_1\times{\bf S}_2){\bf\cdot}d{\bf S}, \label{eq:top2}
\eeqn
is an invariant 2-form. Taking the exterior derivative of $\omega^2$ gives
\beqn
d\omega^2=\nabla{\bf\cdot}({\bf S}_1\times{\bf S}_2)\ d^3x. \label{eq:top3}
\eeqn
In general $\nabla{\bf\cdot}({\bf S}_1\times{\bf S}_2)\neq 0$, and hence
the 3-form $d\omega^2$ has non-zero topological charge. More precisely,
the topological charge for a volume $V=D_3(t)$ is given by the
equivalent expressions:
\beqn
I^q=\int _{D_3(t)} d\omega^2=\int_{\partial D_3(t)}\omega^2
=\int_{\partial D_3(t)}
({\bf S}_1\times{\bf S}_2){\bf\cdot}d{\bf S}. \label{eq:top3a}
\eeqn
Thus $I^2$ is zero if the normal component of ${\bf S}_1\times{\bf S}_2$
is zero on the boundary $\partial D_3(t)$ of the volume $D_3(t)$ of the
region of interest.
\begin{example}{\bf 1.}
For compressible ideal fluid flows:
\beqn
\omega^1_1=\nabla{S}{\bf\cdot}d{\bf x}, \quad 
\omega_2^1=({\bf u}-\nabla\phi+r\nabla S){\bf\cdot}d{\bf x}
\equiv {\bf w}{\bf\cdot} d{\bf x}, \label{eq:top4}
\eeqn
are invariant 1-forms advected with the flow. 
${\bf w}{\bf\cdot}d{\bf x}$ is an invariant advected 1-form, 
where ${\bf u}=\nabla\phi-r\nabla S-\lambda\nabla\mu$ is the  Clebsch
representation for the fluid velocity ${\bf u}$. The two-form $\omega^2$
with properties:
\beqnar
\omega^2&=&\omega^1_1\wedge\omega^1_2=\nabla S\times({\bf u}
-\nabla\phi){\bf\cdot}d{\bf S}, \nonumber\\
d\omega^2&=&\nabla{\bf\cdot}[\nabla S\times({\bf u}-\nabla\phi)] d^3x,
\label{eq:top5}
\eeqnar
is an advected invariant 2-form. Using the identity
\beqn
\nabla{\bf\cdot}({\bf E}\times{\bf A})={\bf A}{\bf\cdot}
\nabla\times{\bf E}-{\bf E}{\bf\cdot}\nabla\times{\bf A}, \label{eq:top6}
\eeqn
with ${\bf E}=\nabla S$ and ${\bf A}={\bf u}-\nabla\phi$ in (\ref{eq:top5})
we obtain:
\beqn
d\omega^2=-\nabla S{\bf\cdot}\nabla\times{\bf u}\ d^3x=-\rho I_e\ d^3x,
\label{eq:top7}
\eeqn
where $I_e$ is the Ertel invariant. In this case, in general
$d\omega^2=\nabla{\bf\cdot}(\rho {\bf J})\ d^3x\neq 0$ where
$\rho {\bf J}=\nabla S\times({\bf u}-\nabla\phi)$. This example shows 
that if $\rho I_e\neq 0$, the Ertel invariant can give rise to 
topological charge in ideal fluid mechanics. 
\end{example}

\begin{example}{\bf 2.}
For ideal MHD,
\beqn
\omega_1^1=\tilde{\bf A}{\bf\cdot}d{\bf x},
\quad \omega_2^1=\nabla S{\bf\cdot}d{\bf x},  \label{eq:top8}
\eeqn
are invariant advected 1-forms. The two form:
\beqn
\omega^2=\omega_1^1\wedge\omega_2^1
=(\tilde{\bf A}\times \nabla S){\bf\cdot}d{\bf S}, \label{eq:top9}
\eeqn
is an advected invariant 2-form, with exterior derivative:
\begin{align}
d\omega^2=&\nabla{\bf\cdot}(\tilde{\bf A}\times\nabla S) d^3x=
\left[\nabla\times\tilde{\bf A}{\bf\cdot}\nabla S-\tilde{\bf A}{\bf\cdot}
(\nabla\times\nabla S)\right] d^3x\nonumber\\
\equiv&({\bf B}{\bf\cdot}\nabla S)d^3x=\rho I_b\ d^3x, \label{eq:top10}
\end{align}
where $I_b={\bf B}{\bf\cdot}\nabla S/\rho$ is an invariant, advected scalar.
In this case $d\omega^2=\nabla{\bf\cdot}(\rho {\bf J})\ d^3x$ 
where $\rho {\bf J}=\tilde{\bf A}\times\nabla S$. If the integral 
$I^2=\int_V  d\omega^2$ is nonzero then it gives a non-zero 
topological charge associated with the scalar 
$I_b={\bf B}{\bf\cdot}\nabla S/\rho$.  
\end{example}

\subsection{The Godbillon Vey Invariant}
In an MHD flow, in which $\tilde{\bf A}{\bf\cdot}\nabla\times\tilde{\bf A}=0$
the magnetic helicity $\tilde{\bf A}{\bf\cdot}{\bf B}=0$. The question 
arises of whether the magnetic field in this case has a non-trivial topology. 
It turns out that the field can still have a non-trivial topology 
if the higher order topological invariant, the Godbillon-Vey invariant 
is non-zero. The same question also arises in ordinary fluid dynamics 
for flows in which ${\bf u}{\bf\cdot}\nabla\times{\bf u}=0$. 
A discussion and derivation of the Godbillon-Vey invariant is given below 
(see also Tur and Janovsky (1993)).

Consider the Pfaffian differential form (1-form) 
$\tilde{\boldsymbol{\omega}}^1_A=\tilde{\bf A}{\bf\cdot}d{\bf x}$, for which 
$d\tilde{\boldsymbol{\omega}}_A^1
=(\nabla\times\tilde{\bf A}){\bf\cdot}d{\bf S}$ and 
\beqn
\tilde{\boldsymbol{\omega}}^1_A\wedge d\tilde{\boldsymbol\omega}_A^1
={\bf A}{\bf\cdot}d{\bf x}\wedge(\nabla\times\tilde{\bf A})
{\bf\cdot}d{\bf S}
=(\tilde{\bf A}{\bf\cdot}\nabla\times{\tilde{\bf A}})\ d^3x. \label{eq:gv1}
\eeqn
The Pfaffian differential equation:
\beqn
\tilde{\boldsymbol{\omega}}_A^1={\tilde{\bf A}}{\bf\cdot}d{\bf x}=0, 
\label{eq:gv2}
\eeqn
determines planes perpendicular to the vector field $\tilde{\bf A}$ at each 
point. For these planes to exist, i.e. for the Pfaffian equation (\ref{eq:gv2}) to have a solution requires that the integrability conditions 
\beqn
\tilde{\boldsymbol\omega}_A^1\wedge d \tilde{\boldsymbol\omega}_A^1
=(\tilde{\bf A}{\bf\cdot}\nabla\times\tilde{\bf A})\ d^3x=0. \label{eq:gv3}
\eeqn
are satisfied. If
\beqn
\tilde{\bf A}{\bf\cdot}\nabla\times{\tilde{\bf A}}=0, \label{eq:gv4}
\eeqn
 the Pfaffian equation (\ref{eq:gv2}) is integrable (e.g. Sneddon (1957)).
 
Tur and Janovsky (1993) discuss the geometric obstruction to integrability 
when $\tilde{\bf A}{\bf\cdot}\nabla\times{\tilde{\bf A}}\neq 0$ in terms 
of non-closure of the integral paths. Note that the helicity or Hopf 
invariant
\beqn
 I^\tau= \int_V\tilde{\boldsymbol{\omega}}^1_A\wedge 
d\tilde{\boldsymbol{\omega}}^1_A
=\int_V \tilde{\bf A}{\bf\cdot}\nabla\times{\tilde{\bf A}}
 d^3x, \label{eq:gv5}
\eeqn
is non-zero 
only if $\tilde{\bf A}{\bf\cdot}\nabla\times{\tilde{\bf A}}
\neq 0$ in some region in the volume $V$ (i.e. $\tilde{\bf A}{\bf\cdot}\nabla\times{\tilde{\bf A}}=0$ throughout the whole of $V$ is not possible). Thus 
$I^\tau\neq 0$ implies $\alpha=\tilde{\bf A}{\bf\cdot}d{\bf x}$ 
is non-integrable in sub-regions of $V$ where $\alpha$ does not change sign.

A natural question (e.g. Tur and Janovsky (1993)), is: given that the 
differential form 
$\tilde{\boldsymbol{\omega}}^1=\tilde{\bf A}{\bf\cdot}d{\bf x}=0$ 
is integrable, and satisfies the integrability condition (\ref{eq:gv3}), 
are there then higher order topological invariants that have non-zero 
topological charge? The answer to this question is yes, there is a higher 
order topological quantity that can be non-zero in this case called 
the Godbillon Vey invariant. It is defined by the equation:
\beqn
I^g =\int_{D^3(t)}\boldsymbol{\eta}{\bf\cdot}
\nabla\times\boldsymbol{\eta}\ d^3x
\quad\hbox{where}\quad \boldsymbol{\eta}=\frac{\tilde{\bf A}\times{\bf B}}
{|\tilde{\bf A}|^2}. \label{eq:gv6}
\eeqn
 where ${\bf B}=\nabla\times\tilde{\bf A}$, and      
${\bf B}{\bf\cdot}{\bf n}=0$ on the boundary $\partial D^3(t)$ of the 
region $D^3(t)$ with outward normal ${\bf n}$.
$I^g$  is a topological invariant that is advected with the flow, 
i.e., 
\beqn
\frac{dI^g}{dt}=0, \label{eq:gv6a}
\eeqn
where $d/dt=\partial/\partial t+{\bf u}{\bf\cdot}\nabla$ is the Lagrangian 
time derivative moving with the flow. 
 It is important to note that 
the Godbillon Vey invariant (\ref{eq:gv6}) only applies to zero helicity 
flows for which $\tilde{\bf A}{\bf\cdot}\nabla\times\tilde{\bf A}=0$.

In (\ref{eq:gv6}) $\boldsymbol{\eta}$ is defined by the integrability 
 equation:
\beqn
d\tilde{\boldsymbol{\omega}}_A^1=\boldsymbol{\omega}_\eta^1\wedge 
\tilde{\boldsymbol{\omega}}^1_A, \label{eq:integr1}
\eeqn
where
\beqn
\tilde{\boldsymbol{\omega}}_A^1=\tilde{\bf A}{\bf\cdot}d{\bf x},\quad\hbox{and} 
\quad \boldsymbol{\omega}_\eta^1=\boldsymbol{\eta}{\bf\cdot}d{\bf x}, 
\label{eq:integr2}
\eeqn
are 1-forms. Taking the exterior derivative of 
$\tilde{\boldsymbol{\omega}}_A^1$ 
and using it in (\ref{eq:integr1}) we obtain the equivalent flux equation:
\beqn
(\nabla\times\tilde{\bf A}){\bf\cdot}d{\bf S}=(\boldsymbol{\eta}\times 
\tilde{\bf A}){\bf\cdot}d{\bf S}\quad \hbox{or}\quad 
\nabla\times\tilde{\bf A}=\boldsymbol{\eta}\times\tilde{\bf A}. 
\label{eq:integr3}
\eeqn
From (\ref{eq:integr3}) we obtain:
\beqn
\tilde{\bf A}\times(\nabla\times\tilde{\bf A})
=\tilde{\bf A}\times(\boldsymbol{\eta}\times\tilde{\bf A})
=(\tilde{\bf A}{\bf\cdot}\tilde{\bf A})\boldsymbol{\eta}
-(\tilde{\bf A}{\bf\cdot}\boldsymbol{\eta})\tilde{\bf A}. \label{eq:integr4}
\eeqn
The general solution of (\ref{eq:integr4}) for $\boldsymbol{\eta}$ is:
\beqn
\boldsymbol{\eta}=\frac{1}{|\tilde{\bf A}|^2}\left(\tilde{\bf A}\times {\bf B}
+\boldsymbol{\eta}{\bf\cdot}\tilde{\bf A}\tilde{\bf A}\right)\label{eq:integr5}
\eeqn
By dropping the arbitrary component of $\boldsymbol{\eta}$ parallel to 
$\tilde{\bf A}$ we obtain the solution (\ref{eq:gv6}) 
for $\boldsymbol{\eta}$. 

A derivation of the Godbillon Vey invariant (\ref{eq:gv6}) and 
and the invariance equation (\ref{eq:gv6a}) for $I^g$  
(see also   
Tur and Janovsky (1993)) is outlined below.

\begin{proof}{\em of Godbillon Vey formula (\ref{eq:gv6a})}

The Frobenius integrability condition (\ref{eq:gv3}) is satisfied 
if there exists a 1-form $\boldsymbol{\omega}^1_{\eta}$ such that 
\beqn
d\tilde{\boldsymbol{\omega}}_A^1=\boldsymbol{\omega}_\eta^1\wedge 
\tilde{\boldsymbol{\omega}}^1_A,  \label{eq:gv7}
\eeqn
Note that 
\beqn
\tilde{\boldsymbol{\omega}}^1_A\wedge d\tilde{\boldsymbol{\omega}}_A^1
=\tilde{\boldsymbol{\omega}}^1_A\wedge (\boldsymbol{\omega}^1_\eta
\wedge\tilde{\boldsymbol{\omega}}^1_A)=-\tilde{\boldsymbol{\omega}}^1_A\wedge
\tilde{\boldsymbol{\omega}}^1_A\wedge\boldsymbol{\omega}^1_\eta=0, 
\label{eq:gv8}
\eeqn
where we used the associative and anti-symmetry properties of the $\wedge$ 
operation. Equation (\ref{eq:gv7}) ensures $d\tilde{\boldsymbol{\omega}}^1_A=0$ 
whenever $\tilde{\boldsymbol{\omega}}^1_A=0$. The condition 
$d\tilde{\boldsymbol{\omega}}^1_A=0$ implies by the Poincar\'e Lemma 
that there exist a $0$-form $\Phi$ such that 
$\tilde{\boldsymbol{\omega}}^1_A=d\Phi$. The Pfaffian 
equation $\tilde{\boldsymbol{\omega}}^1_A=\tilde{\bf A}{\bf\cdot}d{\bf x}=0$ 
is then satisfied by $\Phi(x,y,z)=const.$. Equation (\ref{eq:gv7}) implies
that the set of forms $\{\tilde{\boldsymbol{\omega}}^1_A,
d\tilde{\boldsymbol{\omega}}^1_A\}$ 
is a closed ideal of differential forms which are in involution according 
to Cartan's theory of differential equations 
(e.g. Harrison and Estabrook (1971), i.e. the equations 
$\tilde{\boldsymbol{\omega}}^1_A=0$ are integrable and satisfy the 
integrability conditions (\ref{eq:gv3})).  Equations (\ref{eq:gv7}) are similar
to the Maurer Cartan equations, which are differentiability conditions 
in differential geometry. 

We require that $d\tilde{\boldsymbol{\omega}}^1_A$ 
is advected with the flow, i.e. 
\beqn
\left(\derv{t}+{\cal L}_{\bf u}\right)d\tilde{\boldsymbol{\omega}}^1_A\equiv 
\left(\derv{t}+{\cal L}_{\bf u}\right)\left(\boldsymbol{\omega}^1_\eta
\wedge\tilde{\boldsymbol{\omega}}^1_A\right)=0. \label{eq:gv9}
\eeqn
Expanding (\ref{eq:gv9}) using the properties of the 
Lie derivative ${\cal L}_{\bf u}$ gives:
\beqn 
\left[\left(\derv{t}+{\cal L}_{\bf u}\right)
\tilde{\boldsymbol{\omega}}^1_\eta\right]
\wedge\tilde{\boldsymbol{\omega}}^1_A
+\boldsymbol{\omega}^1_\eta\wedge\left[\left(\derv{t}+{\cal L}_{\bf u}\right)
\tilde{\boldsymbol{\omega}}^1_A\right]=0. \label{eq:gv10}
\eeqn
Using (\ref{eq:gv10}) and the condition that $\tilde{\boldsymbol{\omega}}_A^1$ 
is Lie dragged with the flow (\ref{eq:gv10}) simplifies to:
\beqn
\left[\left(\derv{t}+{\cal L}_{\bf u}\right)
\boldsymbol{\omega}^1_\eta\right]
\wedge \tilde{\boldsymbol{\omega}}^1_A=0, \label{eq:gv11}
\eeqn

Equation (\ref{eq:gv11}) is satisfied if 
\beqn
\left(\derv{t}+{\cal L}_{\bf u}\right)\boldsymbol{\omega}^1_\eta=\alpha
\tilde{\boldsymbol{\omega}}^1_A, \label{eq:gv12}
\eeqn
Equation (\ref{eq:gv12}) can also be written in the form:
\beqn
\deriv{\boldsymbol{\eta}}{t}-{\bf u}\times(\nabla\times\boldsymbol{\eta})
+\nabla({\bf u}{\bf\cdot}\boldsymbol{\eta})=\alpha\tilde{\bf A}. 
\label{eq:gv12a}
\eeqn
Taking the scalar product of (\ref{eq:gv12a}) with $\tilde{\bf A}$ gives:
\beqn
\alpha|\tilde{\bf A}|^2=\tilde{\bf A}{\bf\cdot}
\left[\deriv{\boldsymbol{\eta}}{t}-{\bf u}
\times(\nabla\times\boldsymbol{\eta})
+\nabla({\bf u}{\bf\cdot}\boldsymbol{\eta})\right]. \label{eq:gv12b}
\eeqn
An alternative expression for $\alpha$ can be obtained by noting that 
$\tilde{\bf A}{\bf\cdot}d{\bf x}$ is Lie dragged with the flow. Thus, 
$\tilde{\bf A}$ satisfies (\ref{eq:2.30a}), and hence:
\beqn
0=\boldsymbol{\eta}{\bf\cdot}
\left[\deriv{\tilde{\bf A}}{t}
-{\bf u}\times(\nabla\times\tilde{\bf A})
+\nabla({\bf u}{\bf\cdot}\tilde{\bf A})\right]. 
\label{eq:gv12c}
\eeqn
Noting that $\tilde{\bf A}{\bf\cdot}\boldsymbol{\eta}=
\tilde{\bf A}{\bf\cdot}(\tilde{\bf A}\times {\bf B}/|\tilde{\bf A}|^2)=0$ 
and adding (\ref{eq:gv12b}) 
and (\ref{eq:gv12c}) we obtain:
\beqn
\alpha=\frac{1}{|\tilde{\bf A}|^2}
\left\{ \tilde{\bf A}{\bf\cdot}[{\bf u\cdot}\nabla\boldsymbol{\eta}
+(\nabla{\bf u})^T{\bf\cdot}\boldsymbol{\eta}]
+\boldsymbol{\eta}{\bf\cdot}[{\bf u}{\bf\cdot}\nabla\tilde{\bf A}
+(\nabla{\bf u})^T{\bf\cdot}\tilde{\bf A}]\right\}. \label{eq:gv12d}
\eeqn

Next we investigate if the 3-form:
\beqn
\boldsymbol{\omega}^3_\eta=\boldsymbol{\omega}^1_\eta\wedge 
d\boldsymbol{\omega}^1_\eta, \label{eq:gv13}
\eeqn
is an advected (Lie dragged) 3-form. We find:
\beqn
\left(\derv{t}+{\cal L}_{\bf u}\right)\boldsymbol{\omega}^3_\eta
=-d\left(\alpha d\tilde{\boldsymbol{\omega}}^1_A\right).
\label{eq:gv14}
\eeqn
To derive (\ref{eq:gv14}) first note that
\begin{align}
\left(\derv{t}+{\cal L}_{\bf u}\right)\boldsymbol{\omega}^3_\eta
=&\left[\left(\derv{t}+{\cal L}_{\bf u}\right)\boldsymbol{\omega}^1_\eta\right]
\wedge d\boldsymbol{\omega}^1_\eta
+\boldsymbol{\omega}^1_\eta\wedge
\left[\left(\derv{t}+{\cal L}_{\bf u}\right)
\wedge d\boldsymbol{\omega}^1_\eta\right]
\nonumber\\
=&\alpha\boldsymbol{\omega}_A^1\wedge d\boldsymbol{\omega}_\eta^1+
\boldsymbol{\omega}^1_\eta\wedge 
d\left[\left(\derv{t}+{\cal L}_{\bf u}\right)
\boldsymbol{\omega}^1_\eta\right].
\label{eq:gv14a}
\end{align}
Next we use the fact that 
$d\boldsymbol{\omega}^1_\eta\wedge\tilde{\boldsymbol{\omega}}^1_A=0$ 
which follows by noting 
\beqn
d\left(d\tilde{\boldsymbol{\omega}}^1_A\right)=0\equiv d\left(\boldsymbol{\omega}^1_\eta\wedge
\tilde{\boldsymbol{\omega}}^1_A\right)=d\boldsymbol{\omega}^1_\eta\wedge 
\tilde{\boldsymbol{\omega}}^1_A-\boldsymbol{\omega}^1_\eta\wedge 
d\tilde{\boldsymbol{\omega}}^1_A, \label{eq:gv15}
\eeqn
and that 
$\boldsymbol{\omega}^1_\eta\wedge d\tilde{\boldsymbol{\omega}}^1_A=0$ by 
(\ref{eq:gv7}). Thus, 
\beqn
\left(\derv{t}+{\cal L}_{\bf u}\right)
\boldsymbol{\omega}^3_\eta
=\boldsymbol{\omega}^1_\eta
\wedge d[\alpha\boldsymbol{\omega}_A^1]
\equiv -d\left[\boldsymbol{\omega}_\eta^1
\wedge\alpha\boldsymbol{\omega}_A^1\right]. \label{eq:gv15a}
\eeqn
which reduces to (\ref{eq:gv14}).

Next we consider the Godbillon Vey integral:
\beqn
I^g =\int \boldsymbol{\omega}^3_\eta
=\int \boldsymbol{\omega}^1_\eta\wedge d\boldsymbol{\omega}^1_\eta 
\equiv\int_{D^3(t)} 
\boldsymbol{\eta}{\bf\cdot}\nabla\times\boldsymbol{\eta}\ d^3x. 
\label{eq:gv16}
\eeqn
Using (\ref{eq:gv14}) gives:
\beqn
\deriv{I^g}{t}=\int\deriv{\boldsymbol{\omega}^3_\eta}{t}
=\int\left[-{\cal L}_{\bf u}\left(\boldsymbol{\omega}^3_\eta\right) 
-d\left(\alpha d\tilde{\boldsymbol{\omega}}^1_A\right)\right]. \label{eq:gv17}
\eeqn
However, using Cartan's magic formula gives 
\beqn
{\cal L}_{\bf u}\left(\boldsymbol{\omega}^3_\eta\right)
=d\left({\bf u}\contr \boldsymbol{\omega}^3_\eta\right)+{\bf u}\contr d\boldsymbol{\omega}^3_\eta= d\left({\bf u}\contr \boldsymbol{\omega}^3_\eta\right), \label{eq:gv18}
\eeqn
(note $\boldsymbol{\omega}^3_\eta$ is a 3-form and hence 
$d \boldsymbol{\omega}^3_\eta=0$). From (\ref{eq:gv18}) and (\ref{eq:gv17}) 
we obtain:
\beqn
\deriv{I^g}{t}=\int_{D^3(t)}-d\left({\bf u}\contr \boldsymbol{\omega}^3_\eta
+\alpha d\tilde{\boldsymbol{\omega}}^1_A\right)
=-\int_{\partial D^3(t)}
\left({\bf u}\contr \boldsymbol{\omega}^3_\eta
+\alpha d\tilde{\boldsymbol{\omega}}^1_A\right), \label{eq:gv18a}
\eeqn
Writing
\beqn
\psi=\boldsymbol{\eta}{\bf\cdot}\nabla\times\boldsymbol{\eta}, \label{eq:gv19}
\eeqn
(\ref{eq:gv18a}) can be written in the form:
\beqnar
\int_{D^3(t)}\deriv{\psi}{t} \ d^3x&=&
-\int\left\{{\bf u}\contr\left(\boldsymbol{\omega}^1_\eta\wedge 
d\boldsymbol{\omega}_\eta^1\right)+\alpha 
d(\tilde{\bf A}{\bf\cdot}d{\bf x})\right\} \nonumber\\
&=&-\int\left\{{\bf u}\contr\left[(\boldsymbol{\eta}{\bf\cdot}
d{\bf x})\wedge(\nabla\times\boldsymbol{\eta}){\bf\cdot}d{\bf S}\right]
+\alpha(\nabla\times\tilde{\bf A}){\bf\cdot}d{\bf S}\right\}\nonumber\\
&=&-\int\left[{\bf u}\contr(\boldsymbol{\eta}{\bf\cdot}
\nabla\times\boldsymbol{\eta})d^3x
+\alpha{\bf B}{\bf\cdot}d{\bf S}\right]
\nonumber\\
&=&-\int_{\partial D^3(t)} \left[\psi {\bf u}{\bf\cdot}d{\bf S}
+\alpha{\bf B}{\bf\cdot}d{\bf S}\right]\nonumber\\
&=&-\int_{D^3(t)}\nabla{\bf\cdot}({\bf u}\psi+\alpha {\bf B})\ d^3x. 
\label{eq:gv20}
\eeqnar
Equation (\ref{eq:gv20}) implies the 
conservation law:
\beqn
\deriv{\psi}{t}+\nabla{\bf\cdot}({\bf u}\psi+\alpha {\bf B})=0. \label{eq:gv21}
\eeqn
where $\alpha$ is given in (\ref{eq:gv12d}).

Integrating the continuity equation (\ref{eq:gv21}) for $\psi$ 
over the volume $D^3(t)$, and using the results 
\beqn
\psi d^3x=\psi({\bf x}_0)d^3x_0,\quad d^3x=Jd^3x_0,\quad \psi J=\psi_0(x_0), 
\quad \frac{d\ln J}{dt}=\nabla{\bf\cdot}{\bf u},
\label{eq:gv22}
\eeqn
from Lagrangian fluid mechanics where $J=\det(x_{ij})$ is the Jacobian 
determinant of $x_{ij}=\partial x^i/\partial x_0^j$ of the Lagrangian 
map relating the Eulerian position coordiante ${\bf x}$ and the Lagrangian
label ${\bf x}_0$ where ${\bf x}={\bf x}_0$ at $t=0$, we obtain:
\beqnar
0&=&\int_{D^3(t)}\left[\deriv{\psi}{t}+\nabla{\bf\cdot}({\bf u}\psi
+\alpha {\bf B})\right]
\ d^3x
\nonumber\\
&=&\int_{D^3(t)} \left[\deriv{\psi}{t}+\left(\psi\frac{d\ln J}{dt}
+{\bf u}{\bf\cdot}\nabla\psi\right)\right]\  Jd^3x_0\nonumber\\
&=&\int_{D^3(t)}\ \left[J\frac{d\psi}{dt}+\psi\frac{dJ}{dt}\right]\ d^3x_0
\nonumber\\
&=&\int_{D^3(t)} \left[\frac{d\psi}{dt} d^3x+\psi\frac{d}{dt}(d^3x)\right].
\label{eq:gv23}
\eeqnar
In the second line in (\ref{eq:gv23}) there is no contribution from 
the $\alpha{\bf B}$ term, because if we apply Gauss's theorem 
$\nabla{\bf\cdot}(\alpha {\bf B})d^3x\to \alpha {\bf B}{\bf\cdot}d{\bf S}
=\alpha{\bf B}{\bf\cdot}\tilde{\bf A} dS/|\tilde{\bf A}|=0$ 
and because ${\bf B}{\bf\cdot}\tilde{\bf A}=0$ is 
the integrability condition for  
$\tilde{\bf A}{\bf\cdot}d{\bf x}=0$. 
The last integral in (\ref{eq:gv23}) can be recognized as $dI^g/dt$. 
Thus, (\ref{eq:gv23}) implies the Lagrangian conservation law:
\beqn
\frac{dI^g}{dt}=0. \label{eq:gv24}
\eeqn
Thus $I^g$ is a constant moving with the flow.  This completes the proof 
of (\ref{eq:gv6a}). 
\end{proof}

\section{Hamiltonian Approach}

In this section we discuss the Hamiltonian approach to MHD and gas dynamics. 
In Section 5.1 we give a brief description of a constrained variational 
principle for MHD using Lagrange multipliers to enforce the constraints 
of mass conservation; the entropy advection equation; Faraday's 
equation and the so-called Lin constraint describing in part, the vorticity
of the flow (i.e. Kelvin's theorem). This leads to Hamilton's canonical 
equations in terms of Clebsch potentials. A basic reference 
 is the paper 
by Zakharov and Kuznetsov (1997). The Lagrange multipliers define 
the Clebsch variables, which give a representation for the 
fluid velocity ${\bf u}$.  In Section 5.2 we transform the canonical 
Poisson bracket obtained from the Clebsch variable approach to a 
non-canonical Poisson bracket written in terms of Eulerian 
physical variables (see e.g. Morrison and Greene (1980,1982), Morrison (1982),  
and Holm and Kupershmidt (1983a,b) for more details). In Section 5.3 
we discuss the connection between the Clebsch variable approach and 
Weber transformations. Our main aim 
is to obtain the Clebsch variable evolution equations that follow from 
the variational principle. We use these evolution equations and Clebsch 
variables later to obtain nonlocal fluid helicity and cross helicity 
conservation laws in the next section.  

\subsection{Clebsch variables and Hamilton's Equations}

Consider the MHD action (modified by constraints):
\beqn
J=\int\ d^3x\ dt  L,  \label{eq:Clebsch1}
\eeqn
where
\beqnar
L=&&\left\{\frac{1}{2}\rho u^2-\epsilon(\rho, S)-\frac{B^2}{2\mu_0}\right\}
+\phi\left(\deriv{\rho}{t}+\nabla{\bf\cdot}(\rho {\bf u})\right)\nonumber\\
&&+\beta\left(\deriv{S}{t}+{\bf u}{\bf\cdot}\nabla S\right)
+\lambda\left(\deriv{\mu}{t}+{\bf u\cdot}\nabla\mu\right) \nonumber\\
&&+\boldsymbol{\Gamma}{\bf\cdot}\left(\deriv{\bf B}{t}-\nabla\times({\bf u}\times{\bf B})
+{\bf u}(\nabla{\bf\cdot B})\right). \label{eq:Clebsch2}
\eeqnar
The Lagrangian in curly brackets equals the kinetic minus
the potential energy (internal thermodynamic energy plus magnetic energy).
The Lagrange multipliers $\phi$, $\beta$, $\lambda$, 
and $\boldsymbol{\Gamma}$ ensure that the 
mass, entropy, Lin constraint, Faraday equations are satisfied. We do not 
enforce $\nabla{\bf\cdot}{\bf B}= 0$, since we are interested in the 
effect of $\nabla{\bf\cdot}{\bf B}\neq 0$ (which is useful for numerical 
MHD where $\nabla{\bf\cdot}{\bf B}\neq 0$). It is straightforward 
to impose $\nabla{\bf\cdot}{\bf B}=0$ if desired, although some care 
is required in the formulation of the Poisson bracket, to ensure that the 
Jacobi identity is satisfied 
(e.g. Morrison and Greene 1982, {\bf Morrison 1982}).  

 Stationary point conditions for the action are $\delta J=0$.
 $\delta J/\delta {\bf u}=0$ gives the Clebsch representation
for ${\bf u}$:
\beqn
{\bf u}=\nabla\phi-\frac{\beta}{\rho}\nabla S-\frac{\lambda}{\rho}\nabla\mu+
{\bf u}_M\label{eq:Clebsch3}
\eeqn
where
\beqn
 {\bf u}_M=-\frac{(\nabla\times\boldsymbol{\Gamma})\times{\bf B}}{\rho}
-\boldsymbol{\Gamma}\frac{\nabla{\bf\cdot B}}{\rho}, \label{eq:Clebsch4}
\eeqn
is magnetic contribution to ${\bf u}$.
 Setting $\delta J/\delta\phi$, $\delta J/\delta\beta$,
$\delta J/\delta \lambda$, $\delta J/\delta\boldsymbol{\Gamma}$ consecutively 
equal to zero gives the mass, entropy advection, Lin constraint, 
and Faraday (magnetic flux conservation) constraint
equations:
\beqnar
&&\rho_t+\nabla{\bf\cdot}(\rho {\bf u})=0,\nonumber\\
&&S_t+{\bf u}{\bf\cdot}\nabla S=0,\nonumber\\
&&\mu_t+{\bf u\cdot}\nabla\mu=0, \nonumber\\
&&{\bf B}_t-\nabla\times({\bf u}\times{\bf B})+{\bf u}(\nabla{\bf\cdot B})=0.
\label{eq:Clebsch5}
\eeqnar

 Setting $\delta J/\delta\rho$, $\delta J/\delta S$, $\delta J/\delta\mu$,
$\delta J/\delta {\bf B}$ equal to zero gives evolution equations 
for the Clebsch potentials $\phi$, $\beta$ $\lambda$ and $\boldsymbol{\Gamma}$
as:
\beqnar
&&-\left(\deriv{\phi}{t}+{\bf u}{\bf\cdot}\nabla\phi\right)
+\frac{1}{2} u^2-h=0, \label{eq:Clebsch6}\\
&&\deriv{\beta}{t}+\nabla{\bf\cdot}(\beta {\bf u})+\rho T=0, 
\label{eq:Clebsch7}\\
&&\deriv{\lambda}{t}+\nabla{\bf\cdot}(\lambda {\bf u})=0, 
\label{eq:Clebsch8}\\
&&\deriv{\boldsymbol{\Gamma}}{t}
-{\bf u}\times (\nabla\times\boldsymbol{\Gamma})
+\nabla(\boldsymbol{\Gamma}{\bf\cdot u})+\frac{\bf B}{\mu_0}=0. 
\label{eq:Clebsch9}
\eeqnar
Equation (\ref{eq:Clebsch6}) is related to Bernoulli's equation for potential 
flow.The $\nabla(\boldsymbol{\Gamma}{\bf\cdot u})$ term 
in (\ref{eq:Clebsch9}) is associated with
$\nabla{\bf\cdot}{\bf B}\neq 0$.
 Taking the curl of (\ref{eq:Clebsch9})  gives:
\beqn
\deriv{\tilde{\boldsymbol{\Gamma}}}{t}
-\nabla\times({\bf u}\times {\tilde{\boldsymbol{\Gamma}}})
=-\frac{\nabla\times{\bf B}}{\mu_0} \quad\hbox{where}
\quad \tilde{\boldsymbol{\Gamma}}=\nabla\times\boldsymbol{\Gamma}.  \label{eq:Clebsch10}
\eeqn

 Equations (\ref{eq:Clebsch6})-(\ref{eq:Clebsch10}) can be written in the form:
\beqnar
&&\frac{d\phi}{dt}=\frac{1}{2}u^2-h,\quad 
\frac{d}{dt}\left(\frac{\beta}{\rho}\right)=-T, \nonumber\\
&&\frac{d}{dt}\left(\lambda d^3x\right)=0\quad \hbox{or}
\quad \frac{d}{dt}\left(\frac{\lambda}{\rho}\right)=0, \nonumber\\
&&\frac{d}{dt}(\boldsymbol{\Gamma}{\bf\cdot}d{\bf x})
=-\frac{{\bf B}{\bf\cdot}d{\bf x}}{\mu_0},\quad
\frac{d}{dt}(\tilde{\boldsymbol{\Gamma}}{\bf\cdot}d{\bf S})
=-{\bf J}{\bf\cdot}d{\bf S}. \label{eq:Clebsch11}
\eeqnar
where $d/dt=\partial/\partial t+{\bf u}{\bf\cdot}\nabla$,
is the Lagrangian time
derivative following the flow and ${\bf J}=\nabla\times{\bf B}/\mu_0$ 
is the current. 

 Introduce the Hamiltonian functional:
\beqn
{\cal H}=\int H d^3x\quad\hbox{where}\quad H=\frac{1}{2}\rho u^2+\epsilon(\rho,S)+\frac{B^2}{2\mu_0}. \label{eq:H1}
\eeqn
 Substitute the Clebsch
expansion (\ref{eq:Clebsch3})-(\ref{eq:Clebsch4}) for ${\bf u}$
in (\ref{eq:H1}). Evaluating the  variational derivatives
of ${\cal H}$ gives Hamilton's equations:
\begin{align}
&\deriv{\rho}{t}=\frac{\delta{\cal H}}{\delta \phi},
\quad \deriv{\phi}{t}=-\frac{\delta {\cal H}}{\delta\rho}, \quad
\deriv{S}{t}=\frac{\delta{\cal H}}{\delta\beta},\quad
\deriv{\beta}{t}=-\frac{\delta{\cal H}}{\delta S}, \nonumber\\
&\deriv{\mu}{t}=\frac{\delta{\cal H}}{\delta{\lambda}}, \quad
\deriv{\lambda}{t}=-\frac{\delta{\cal H}}{\delta\mu}, \quad
\deriv{\bf B}{t}=\frac{\delta{\cal H}}{\delta\boldsymbol{\Gamma}}, \quad
\deriv{\boldsymbol{\Gamma}}{t}=-\frac{\delta{\cal H}}{\delta{\bf B}}. 
\label{eq:H3} 
\end{align}
Here  $\{\rho,\phi\}$, $\{S,\beta\}$,
$\{\mu,\lambda\}$, $\{{\bf B},\boldsymbol{\Gamma}\}$
are canonically conjugate variables.

The canonical Poisson bracket is:
\begin{align}
\{F,G\}=&\int d^3x\ \biggl(\frac{\delta F}{\delta\rho}
\frac{\delta G}{\delta \phi}-\frac{\delta F}{\delta \phi}
\frac{\delta G}{\delta\rho}
+\frac{\delta F}{\delta{\bf B}}
{\bf\cdot}\frac{\delta G}{\delta \boldsymbol{\Gamma}}
-\frac{\delta F}{\delta \boldsymbol{\Gamma}}
{\bf\cdot}\frac{\delta G}{\delta{\bf B}}\nonumber\\
&\quad +\frac{\delta F}{\delta S}
\frac{\delta G}{\delta \beta}-\frac{\delta F}{\delta\beta}
\frac{\delta G}{\delta S}
+\frac{\delta F}{\delta\mu}
\frac{\delta G}{\delta\lambda}-\frac{\delta F}{\delta\lambda}
\frac{\delta G}{\delta\mu}
\biggr). \label{eq:H4}
\end{align}

It is straightforward to verify that the canonical Poisson bracket (\ref{eq:H4})
satisfies the linearity, skew symmetry and Jacobi identity necessary
for a Hamiltonian system (i.e. the Poisson bracket defines a Lie algebra).
\subsection{Non-Canonical Poisson Brackets}
{\bf Morrison and Greene (1980,1982)  introduced  non-canonical 
Poisson brackets for MHD for the case $\nabla{\bf\cdot}{\bf B}=0$. 
Morrison and Greene (1982) and Morrison (1982) discuss the form of the 
Poisson bracket if $\nabla{\bf\cdot}{\bf B}\neq 0$.  Morrison (1982) 
discusses the proof of the Jacobi identity. Holm and Kupershmidt (1983) 
point out that the Poisson bracket has the form expected for a semi-direct 
product Lie algebra, for which the Jacobi identity is automatically satisfied. 
Chandre et al. (2013) discuss the $\nabla{\bf\cdot}{\bf B}=0$ constraint 
using Dirac's 
method of constraints and the Dirac bracket.}  

Introduce the new variables:
\beqn
{\bf M}=\rho {\bf u}, \quad \sigma=\rho S, \label{eq:H5}
\eeqn
The formulae for the transformation of 
variational derivatives in the old variables 
$(\rho,\phi,S,\beta,{\bf B},\boldsymbol{\Gamma})$ in terms of the 
new variables $(\rho,\sigma,{\bf B},{\bf M})$ are:
\begin{align}
&\frac{\delta F}{\delta\rho}=\frac{\delta F}{\delta\rho}+S\frac{\delta F}
{\delta\sigma}
+\frac{\delta F}{\delta{\bf M}}{\bf\cdot}\nabla\phi,\quad
\frac{\delta F}{\delta\phi}=
-\nabla{\bf\cdot}\left(\rho\frac{\delta F}{\delta{\bf M}}\right),\nonumber\\
&\frac{\delta F}{\delta S}=\rho\frac{\delta F}{\delta\sigma}
+\nabla{\bf\cdot}\left(\beta\frac{\delta F}{\delta {\bf M}}\right),\quad
\frac{\delta F}{\delta \beta}
=-\frac{\delta F}{\delta {\bf M}}{\bf\cdot}\nabla S, \nonumber\\
&\frac{\delta F}{\delta {\bf B}}=\left[\frac{\delta F}{\delta{\bf B}}
+\nabla\left(\frac{\delta F}{\delta {\bf M}}\right){\bf\cdot}
\boldsymbol{\Gamma}
+\frac{\delta F}{\delta {\bf M}}{\bf\cdot}\nabla\boldsymbol{\Gamma}\right]
\equiv
\frac{\delta F}{\delta{\bf B}}
+\nabla\left(\boldsymbol{\Gamma}{\bf\cdot}
\frac{\delta F}{\delta {\bf M}}\right)
+(\nabla\times\boldsymbol{\Gamma})\times 
\frac{\delta F}{\delta {\bf M}}, \nonumber\\
&\frac{\delta F}{\delta\boldsymbol{\Gamma}}=\left[{\bf B}{\bf\cdot}\nabla
\left(\frac{\delta F}{\delta {\bf M}}\right) 
-\nabla{\bf\cdot}\left(\frac{\delta F}{\delta{\bf M}}\right){\bf B}
-\frac{\delta F}{\delta{\bf M}}{\bf\cdot}\nabla{\bf B}\right]
\equiv \nabla\times\left(\frac{\delta F}{\delta {\bf M}}\times{\bf B}\right)
-\frac{\delta F}{\delta {\bf M}}(\nabla{\bf\cdot}{\bf B}),\nonumber\\
&\frac{\delta F}{\delta\mu}
=\nabla{\bf\cdot}\left(\lambda\frac{\delta F}{\delta{\bf M}}\right), \quad
\frac{\delta F}{\delta\lambda}=-\frac{\delta F}{\delta{\bf M}}
{\bf\cdot}\nabla\mu.  \label{eq:H6}
\end{align}
Note that 
\beqn
{\bf M}=\rho {\bf u}=\rho\nabla\phi-\beta\nabla S
+{\bf B}{\bf\cdot}(\nabla\boldsymbol{\Gamma})^T
-{\bf B}{\bf\cdot}\nabla\boldsymbol{\Gamma}
-\boldsymbol{\Gamma}(\nabla{\bf\cdot}{\bf B}). \label{eq:H7}
\eeqn

Using the transformations (\ref{eq:H6}) in the canonical Poisson 
bracket (\ref{eq:H4})
we obtain the Morrison and Greene (1982) non-canonical  
Poisson bracket:
\begin{align}
\left\{F,G\right\}=&-\int\ d^3x \biggl\{ \rho
\left[\frac{\delta F}{\delta{\bf M}}
{\bf\cdot}\nabla\left(\frac{\delta G}{\delta\rho}\right)
-\frac{\delta G}{\delta{\bf M}}
{\bf\cdot}\nabla\left(\frac{\delta F}{\delta\rho}\right)\right]\nonumber\\
&+\sigma\left[\frac{\delta F}{\delta{\bf M}}
{\bf\cdot}\nabla\left(\frac{\delta G}{\delta\sigma}\right)
-\frac{\delta G}{\delta{\bf M}}
{\bf\cdot}\nabla\left(\frac{\delta F}{\delta\sigma}\right)\right]\nonumber\\
&+{\bf M}{\bf\cdot}\left[\left(\frac{\delta F}{\delta{\bf M}}
{\bf\cdot}\nabla\right)\frac{\delta G}{\delta{\bf M}}
-\left(\frac{\delta G}{\delta{\bf M}}
{\bf\cdot}\nabla\right)\frac{\delta F}{\delta{\bf M}}\right]\nonumber\\
&+{\bf B}{\bf\cdot}\left[\frac{\delta F}{\delta{\bf M}}
{\bf\cdot}\nabla\left(\frac{\delta G}{\delta{\bf B}}\right)
-\frac{\delta G}{\delta{\bf M}}
{\bf\cdot}\nabla\left(\frac{\delta F}{\delta{\bf B}}\right)\right]\nonumber\\
&+{\bf B}{\bf\cdot}\left[
\left(\nabla\frac{\delta F}{\delta {\bf M}}\right){\bf\cdot}
\frac{\delta G}{\delta{\bf B}}
-\left(\nabla\frac{\delta G}{\delta{\bf M}}\right){\bf\cdot}
\frac{\delta F}{\delta{\bf B}}\right]
\biggr\}. \label{eq:H9}
\end{align}
The bracket (\ref{eq:H9}) has the Lie-Poisson form and 
satisfies the Jacobi identity for all functionals 
$F$ and $G$ of the physical variables, and in general applies both for 
$\nabla{\bf\cdot}{\bf B}\neq 0$  and $\nabla{\bf\cdot}{\bf B}=0$.

\subsection{Weber Transformations}
 The classical Weber transformation uses the Lagrangian
map: ${\bf x}={\bf x}({\bf x}_0,t)$ to integrate the Eulerian
momentum equation to get the 
Clebsch representation for ${\bf u}$.
 The Eulerian momentum conservation equation  can be written as:
\beqn
\derv{t}(\rho {\bf u})+\nabla{\bf\cdot}\left[\rho{\bf u}\otimes{\bf u}
+p{\sf I}
+\left(\frac{B^2}{2\mu_0}{\sf I}
-\frac{{\bf B}\otimes{\bf B}}{\mu_0}\right)\right]=0, \label{eq:wt1}
\eeqn
or as:
\beqn 
\frac{d{\bf u}}{dt}=T\nabla S-\nabla h+\frac{{\bf J}\times {\bf B}}{\rho}
+{\bf B}
\frac{\nabla{\bf\cdot}{\bf B}}{\mu_0\rho}. \label{eq:wt2}
\eeqn
Use:
\beqnar
&&\frac{d{\bf u}}{dt}
=\deriv{\bf u}{t}+{\boldsymbol\omega}\times{\bf u}
+\nabla\left(\frac{1}{2}|{\bf u}|^2\right)\quad
\hbox{where}\quad  {\boldsymbol\omega}=\nabla\times{\bf u},\nonumber\\
&&\frac{d}{dt}\left({\bf u\cdot}d{\bf x}\right)
=\left[\deriv{\bf u}{t}+{\boldsymbol\omega}\times{\bf u}
+\nabla\left(|{\bf u}|^2\right)\right]{\bf\cdot}d{\bf x}, \label{eq:wt3}
\eeqnar
to get:
\beqn
\frac{d}{dt}\left({\bf u\cdot}d{\bf x}\right)=
\left[T\nabla S
+\nabla \left(\frac{1}{2} |{\bf u}|^2-h\right)
+\frac{{\bf J}\times {\bf B}}{\rho}
+{\bf B}
\frac{\nabla{\bf\cdot}{\bf B}}{\mu_0\rho}\right]{\bf\cdot}d{\bf x}.
\label{eq:wt4}
\eeqn

 On  the right-hand side (RHS) of (\ref{eq:wt4}) for the magnetic terms we 
use:
\beqnar
&&\frac{d}{dt}
\left[\left(\frac{(\nabla\times\boldsymbol{\Gamma})\times {\bf B}}{\rho}\right)
{\bf\cdot}d{\bf x}\right]
= -\left(\frac{{\bf J}\times{\bf B}}{\rho}\right){\bf\cdot}d{\bf x}
\label{eq:wt5}\\
&&\frac{d}{dt}
\left[\left(\frac{\nabla{\bf\cdot B}}{\rho}\right) \boldsymbol{\Gamma}
{\bf\cdot}d{\bf x}\right]
=-\left(\frac{\nabla{\bf\cdot B}}{\rho}\right)
\frac{\bf B}{\mu_0}{\bf\cdot}d{\bf x}.
\label{eq:wt6}
\eeqnar
 On the right hand side of (\ref{eq:wt4}) for the gas bits we use:
\beqnar
&&\frac{d}{dt}\left(\nabla\phi{\bf\cdot}d{\bf x}\right)=
\nabla\left(\frac{1}{2}\left|{\bf u}\right|^2-h\right){\bf\cdot}d{\bf x},
\nonumber\\
&&\frac{d}{dt}\left(r\nabla S{\bf\cdot}d{\bf x}\right)
=-T\nabla S{\bf\cdot}d{\bf x},
\quad \frac{d}{dt}\left(\tilde{\lambda}\nabla\mu{\bf\cdot}d{\bf x}\right)=0, \nonumber\\
&&\tilde{\lambda}=\frac{\lambda}{\rho}, \quad r=\frac{\beta}{\rho}.
\label{eq:wt7}
\eeqnar
 to obtain the Clebsch
representation ${\bf u}={\bf u}_h+{\bf u}_M$
in (\ref{eq:Clebsch3})-(\ref{eq:Clebsch4}).

\begin{proposition}
Equations (\ref{eq:wt4}) -(\ref{eq:wt7}) imply the Clebsch representation 
${\bf u}={\bf u}_h+{\bf u}_M$ in (\ref{eq:Clebsch3})-(\ref{eq:Clebsch4}).
\end{proposition}

\begin{proof}
Using (\ref{eq:wt5})-(\ref{eq:wt7}) in (\ref{eq:wt4})
gives:
\beqn
\frac{d}{dt}\left({\bf w}{\bf\cdot}d{\bf x}\right)=0, \label{eq:wt8a}
\eeqn
where
\beqn
{\bf w}={\bf u}
-\left(\nabla\phi-r\nabla S-\frac{\nabla\times\boldsymbol{\Gamma}}{\rho}
\times{\bf B}
-\left(\frac{\nabla{\bf\cdot B}}{\rho}
\right)
\boldsymbol{\Gamma}\right). \label{eq:wt8b}
\eeqn

Integration of  (\ref{eq:wt8a}) gives
\beqn
{\bf w}{\bf\cdot}d{\bf x}=f_0({\bf x}_0)^k dx_0^k
\quad\hbox{or}\quad
w^j=f_0({\bf x}_0)^k \partial x_0^k/\partial x^j. \label{eq:wt8c}
\eeqn
Using the  initial data: 
$w^j=f_0({\bf x}_0)^j=f_{00}({\bf x}_0) \partial g_{00}/\partial x_0^j$ 
at $t=0$ gives
\beqn
{\bf w}=-\tilde{\lambda}\nabla\mu, \label{eq:wt8d}
\eeqn
where $\tilde{\lambda}=-f_{00}$ and $\mu=g_{00}$. Equations 
(\ref{eq:wt8b})-(\ref{eq:wt8c})
then give:
\beqn
{\bf u}=\nabla\phi-\tilde{\lambda}\nabla\mu-r\nabla S-\frac{\nabla\times
\boldsymbol{\Gamma}}{\rho}\times{\bf B}-\left(\frac{\nabla{\bf\cdot B}}{\rho}
\right)\boldsymbol{\Gamma}, \label{eq:wt8e}
\eeqn
which is the Clebsch representation (\ref{eq:Clebsch3})-(\ref{eq:Clebsch4}) 
for ${\bf u}$, where $\tilde{\lambda}=\lambda/\rho$.

The proof of (\ref{eq:wt5}) is sketched below.
 Note that ${\bf b}={\bf B}/\rho$ is an advected vector field.
The one form on the LHS of (\ref{eq:wt5}) can be written as
\beqn
\alpha={\bf b}{\mathrel{\lrcorner}}
\left(\tilde{\boldsymbol{\Gamma}}{\bf\cdot}d{\bf S}\right)
=(\tilde{\boldsymbol{\Gamma}}\times{\bf b}){\bf\cdot}d{\bf x}
\equiv [(\nabla\times{\boldsymbol{\Gamma}})
\times {\bf B}/\rho]{\bf\cdot}d{\bf x}.
\label{eq:wt8}
\eeqn
where $\tilde{\boldsymbol{\Gamma}}=\nabla\times{\boldsymbol{\Gamma}}$.
The RHS of (\ref{eq:wt5}) is:
\beqnar
\frac{d\alpha}{dt}=&&\frac{d{\bf b}}{dt}
{\mathrel{\lrcorner}}\tilde{\boldsymbol{\Gamma}}{\bf\cdot}d{\bf S}
+{\bf b}{\mathrel{\lrcorner}}\frac{d}{dt}(\tilde{\boldsymbol{\Gamma}}
{\bf\cdot}d{\bf S})\nonumber\\
=&&0-{\bf b}{\mathrel{\lrcorner}}({\bf J}{\bf\cdot}d{\bf S})
\equiv -\frac{{\bf J}\times {\bf B}}{\rho}{\bf\cdot}d{\bf x}. \label{eq:wt9}
\eeqnar
This establishes (\ref{eq:wt5}). There are similar proofs for
(\ref{eq:wt6}) and (\ref{eq:wt7}).
\end{proof}

\section{Nonlocal Helicity Conservation Laws}

In this section we look again at the  
helicity conservation law (\ref{eq:hf2})
and the cross helicity conservation law (\ref{eq:crh2}). 
 The  helicity conservation law (\ref{eq:hf2}) requires that
the gas or fluid be barotropic (i.e. $p=p(\rho)$ is independent
of the entropy $S$) in order for this conservation law to apply.
Similarly, the cross helicity conservation equation (\ref{eq:crh2}) only
applies, if either (\romannumeral1)\ the gas is barotropic with $p=p(\rho)$
or if (\romannumeral2) ${\bf B}{\bf\cdot}\nabla S=0$, which implies that
the magnetic field lies in the constant entropy surface.
Using Clebsch variables allows
one to obtain analogous nonlocal conservation laws corresponding to the 
helicity in ordinary fluid dynamics, and the cross helicity conservation
law in MHD. 

\begin{proposition}\label{propnl1}
The generalized helicity conservation law in ideal fluid mechanics 
can be written in the form:
\beqn
\derv{t}\left[\boldsymbol{\Omega}{\bf\cdot}({\bf u}+r\nabla S)\right] 
+\nabla{\bf\cdot}\left\{{\bf u}\left[\boldsymbol{\Omega}{\bf\cdot}({\bf u}+r\nabla S)\right]+\boldsymbol{\Omega}\left(h-\frac{1}{2}|{\bf u}|^2\right)\right\}
=0. \label{eq:nl1}
\eeqn
The nonlocal conservation law (\ref{eq:nl1}) depends on the Clebsch variable 
formulation of ideal fluid mechanics in which the fluid velocity ${\bf u}$ 
is given by the equation:
\beqn
{\bf u}=\nabla\phi-r\nabla S-\tilde{\lambda} \nabla\mu, \label{eq:nl2}
\eeqn
where $\phi$, $r$, $S$, $\tilde{\lambda}$, and $\mu$ satisfy the equations:
\begin{align}
&\frac{d\phi}{dt}=\frac{1}{2} |{\bf u}|^2-h, \quad \frac{dr}{dt}=-T, 
\nonumber\\
&\frac{dS}{dt}=\frac{d\tilde{\lambda}}{dt}=\frac{d\mu}{dt}=0, 
\label{eq:nl3}
\end{align}
and $d/dt=\partial/\partial t+{\bf u}{\bf\cdot}\nabla$ is the Lagrangian time 
derivative following the flow. In (\ref{eq:nl1}) the generalized vorticity 
$\boldsymbol{\Omega}$ is defined by the equations:
\begin{align}
&{\bf w}={\bf u}-\nabla\phi+r\nabla S\equiv -\tilde{\lambda}\nabla\mu, 
\label{eq:nl4}\\
&\boldsymbol{\Omega}=\nabla\times{\bf w}=\boldsymbol{\omega}
+\nabla r\times\nabla S, \label{eq:nl5}
\end{align}
where $\boldsymbol{\omega}=\nabla\times{\bf u}$ is the fluid vorticity. 
The one-form $\alpha={\bf w}{\bf\cdot}d{\bf x}$ and the two-form 
$\beta=d\alpha=\boldsymbol{\Omega}{\bf\cdot}d S$ 
are advected invariants, i.e. 
\begin{align}
&\frac{d\alpha}{dt}=\left(\derv{t}+{\cal L}_{\bf u}\right) \alpha\equiv
\left[\deriv{\bf w}{t}-{\bf u}\times(\nabla\times{\bf w}) 
+\nabla({\bf u}{\bf\cdot}{\bf w})\right]{\bf\cdot}d{\bf x}=0, \label{eq:nl6}\\
&\frac{d\beta}{dt}=\left(\derv{t}+{\cal L}_{\bf u}\right) \beta\equiv
\left[\deriv{\boldsymbol{\Omega}} {t}- \nabla\times({\bf u}\times
\boldsymbol{\Omega})\right] {\bf\cdot}d{\bf S}=0. \label{eq:nl7}
\end{align} 
An alternative form of the conservation law (\ref{eq:nl1}) is:
\beqn
\derv{t}\left({\bf u\bf\cdot}\boldsymbol{\Omega}+\beta I_e\right)
+\nabla{\bf\cdot}
\bigl[({\bf u\cdot}\boldsymbol{\Omega}){\bf u}
+(\beta I_e){\bf u}
+\boldsymbol{\Omega}\left(h-\frac{1}{2}|{\bf u}|^2\right)
\bigr]=0, 
\label{eq:nl7a}
\eeqn
where
\beqn
\beta=r\rho, \quad I_e=\frac{\boldsymbol{\omega}{\bf\cdot}\nabla S}{\rho}
\equiv\frac{\boldsymbol{\Omega}{\bf\cdot}\nabla S}{\rho} 
\label{eq:nl7b}
\eeqn
in which $I_e$ is the Ertel invariant and the Clebsch variable $\beta$ 
satisfies the evolution equation:
\beqn
\deriv{\beta}{t}+\nabla{\bf\cdot}(\beta{\bf u})=-\rho T. \label{eq:nl7c}
\eeqn
\end{proposition}

\begin{proof}
Equation (\ref{eq:nl6}) states that 
the one-form $\alpha={\bf w}{\bf\cdot}d{\bf x}$ is Lie dragged by the 
flow. This is proved, by calculating the evolution of the forms 
${\bf u}{\bf\cdot}d{\bf x}$, $-\nabla\phi{\bf\cdot}d{\bf x}$ and 
$r\nabla S{\bf\cdot}d {\bf x}$ moving with the flow (see (\ref{eq:wt8}) 
et seq. for a proof the $d/dt({\bf w}{\bf\cdot}d{\bf x})=0$). Alternatively
one can prove (\ref{eq:nl6}) by noting ${\bf w}
\equiv -\tilde{\lambda}\nabla\mu$ and that $\tilde{\lambda}$ and $\mu$ 
are advected with the flow. It is straightforward to show:
\beqn
\deriv{\bf w}{t}-{\bf u}\times(\nabla\times{\bf w}) 
+\nabla({\bf u}{\bf\cdot}{\bf w})
=-\nabla\mu\left(\frac{\tilde{d\lambda}}{dt}\right)
-\tilde{\lambda} \nabla\left(\frac{d\mu}{dt}\right)=0, 
\label{eq:nl7d}
\eeqn
which verifies (\ref{eq:nl6}).  

Because $\alpha={\bf w}{\bf\cdot}d{\bf x}$ is an advected invariant 
one form, then $\beta=d\alpha=\boldsymbol{\Omega}{\bf\cdot}d{\bf S}$ 
is an advected invariant 2-form and hence by (\ref{eq:liedr3}) 
$\boldsymbol{\Omega}$
satisfies the equation:
\beqn
\deriv{\boldsymbol{\Omega}}{t}
-\nabla\times\left({\bf u}\times\boldsymbol{\Omega}\right)
+(\nabla{\bf\cdot}\boldsymbol{\Omega}){\bf u}=0, \label{eq:nl10b}
\eeqn
where $\nabla{\bf\cdot}\boldsymbol{\Omega}=0$. 

From (\ref{eq:hf5}), the momentum equation for the system may be written 
in the form:
\beqn
\deriv{\bf u}{t}-{\bf u}\times\boldsymbol{\omega}
+\nabla\left(\frac{1}{2}|{\bf u}|^2\right)=T\nabla S-\nabla h. 
\label{eq:nl10c}
\eeqn 
Combining (\ref{eq:nl10b}) and (\ref{eq:nl10c}) gives the equation:
\beqn
\boldsymbol{\Omega}{\bf\cdot}\left[\deriv{\bf u}{t}
-{\bf u}\times\boldsymbol{\omega}
+\nabla\left(h+\frac{1}{2}|{\bf u}|^2\right)-T\nabla S\right]
+{\bf u}{\bf\cdot}\left[\deriv{\boldsymbol{\Omega}}{t}
-\nabla\times({\bf u}\times\boldsymbol{\Omega})\right]=0. \label{eq:nl10d}
\eeqn
Equation (\ref{eq:nl10d}) reduces to: 
\beqn
\derv{t}({\bf u\cdot}\boldsymbol{\Omega})
+\nabla{\bf\cdot}\left[{\bf u}\times({\bf u}\times\boldsymbol{\Omega})\right]
+\nabla{\bf\cdot}
\left[\left(h+\frac{1}{2}|{\bf u}|^2\right)\boldsymbol{\Omega}
\right]=T\boldsymbol{\Omega}{\bf\cdot}\nabla S. \label{eq:nl10e}
\eeqn
Since $\boldsymbol{\Omega}=\boldsymbol{\omega}+\nabla r\times\nabla S$, 
 the right handside of (\ref{eq:nl10e}) may be written as:
\beqn
T\boldsymbol{\Omega}{\bf\cdot}\nabla S
=T\boldsymbol{\omega}{\bf\cdot}\nabla S=\rho T I_e, \label{eq:nl10g}
\eeqn
where $I_e=\boldsymbol{\omega}{\bf\cdot}\nabla S/\rho$ is the Ertel invariant.
Using (\ref{eq:nl7c}) in (\ref{eq:nl10g}) we obtain:
\beqn
T\boldsymbol{\Omega}{\bf\cdot}\nabla S
=-\left[\deriv{\beta}{t}+\nabla{\bf\cdot}(\beta {\bf u})\right] I_e
=-\left[\deriv{(\beta I_e)}{t}+\nabla{\bf\cdot}(\beta I_e {\bf u})\right] 
+\beta\frac{dI_e}{dt}. \label{eq:nl10h}
\eeqn
However $dI_e/dt=0$. Thus, using (\ref{eq:nl10h}) in (\ref{eq:nl10e}) 
we obtain the conservation law:
\beqn
\derv{t}\left({\bf u\bf\cdot}\boldsymbol{\Omega}+\beta I_e\right)
+\nabla{\bf\cdot}
\bigl[({\bf u\cdot}\boldsymbol{\Omega}){\bf u}
+(\beta I_e){\bf u}
+\boldsymbol{\Omega}\left(h-\frac{1}{2}|{\bf u}|^2\right)
\bigr]=0,
\label{eq:nl10i}
\eeqn
which is  (\ref{eq:nl7a}). By noting:
\beqn
\beta I_e=r\rho(\boldsymbol{\omega}{\bf\cdot}\nabla S)/\rho
=r\boldsymbol{\omega}{\bf\cdot}\nabla S
\equiv r\boldsymbol{\Omega}{\bf\cdot}\nabla S , \label{eq:nl10j} 
\eeqn
 (\ref{eq:nl10i}) reduces to (\ref{eq:nl1}). This completes the proof
\end{proof}

\begin{remark}{\bf 1}
Since $\alpha$ and $\beta$ are advected invariants, then the three form
\beqn
\gamma=\alpha\wedge\beta=({\bf w}{\bf\cdot}\boldsymbol{\Omega})\ d^3x, 
\label{eq:nl8}
\eeqn
is also an advected invariant. However 
\beqn
{\bf w}{\bf\cdot}\boldsymbol{\Omega}
=(-\tilde{\lambda}\nabla\mu){\bf\cdot}
(-\nabla\tilde{\lambda}\times\nabla\mu)=0. \label{eq:nl9}
\eeqn
Taking into account (\ref{eq:nl9}) the conservation law (\ref{eq:nl1}) 
can also be written in the form:
\beqn
\derv{t}\left(\boldsymbol{\Omega}{\bf\cdot}\nabla\phi\right) 
+\nabla{\bf\cdot}\left[{\bf u}(\boldsymbol{\Omega}{\bf\cdot}\nabla\phi) 
+\boldsymbol{\Omega}\left(h-\frac{1}{2}|{\bf u}|^2\right)\right]=0. 
\label{eq:nl10}
\eeqn
\end{remark}
\begin{remark}{\bf 2}
The conservation laws (\ref{eq:nl1}), or equivalently (\ref{eq:nl10}) 
is a nonlocal conservation law that involves the nonlocal 
potentials:
\begin{align}
&r({\bf x},t)=-\int_0^t T_0({\bf x}_0,t')\ dt'+r_0({\bf x}_0), 
\label{eq:nl11}\\
&\phi({\bf x},t)=\int_0^t \left(\frac{1}{2}|{\bf u}|^2
-h\right)({\bf x}_0,t')\ dt' 
+\phi_0({\bf x}_0), \label{eq:nl12}
\end{align}
where ${\bf x}={\bf f}({\bf x}_0,t)$ and ${\bf x}_0={\bf f}^{-1}({\bf x},t)$ 
are the Lagrangian map and the inverse Lagrangian map. The temperature 
$T({\bf x},t)=T_0({\bf x}_0,t)$ and $r_0({\bf x}_0)$ and $\phi_0({\bf x}_0)$
are 'integration constants'.
\end{remark}

\begin{remark}{\bf 3}
The conservation law (\ref{eq:nl1}) can also be written in the form:
\beqn
\deriv{D}{t}+\nabla{\bf\cdot}{\bf F}=0, \label{eq:nl13}
\eeqn
where
\begin{align}
D=&\boldsymbol{\omega}{\bf\cdot}{\bf u}
+{\bf u}{\bf\cdot}\nabla r\times\nabla S, \label{eq:nl14}\\
{\bf F}=&{\bf u}(\boldsymbol{\omega}{\bf\cdot}{\bf u})
+\boldsymbol{\omega}
\left(h-\frac{1}{2}
|{\bf u}|^2\right)\nonumber\\
&+{\bf u}\left({\bf u}{\bf\cdot}\nabla r\times\nabla S\right) 
+\left(\nabla r\times\nabla S\right)
\left(h-\frac{1}{2}
|{\bf u}|^2\right). \label{eq:nl15}
\end{align}
For barotropic or constant entropy flows $\nabla S=0$ and the conservation 
law  (\ref{eq:nl1}) reduces to the usual fluid helicity conservation form:
\beqn
\derv{t}\left(\boldsymbol{\omega}{\bf\cdot}{\bf u}\right)
+\nabla{\bf\cdot}
\bigl[{\bf u}(\boldsymbol{\omega}{\bf\cdot}{\bf u})
+\boldsymbol{\omega}
\left(h-\frac{1}{2}
|{\bf u}|^2\right)\bigr]=0 \label{eq:nl16}
\eeqn
where $h=(\varepsilon+p)/\rho$ is the entropy of the gas. 
\end{remark}

\begin{proposition}\label{propnl2}
The generalized cross helicity conservation law in MHD can be written in the 
form:
\beqn
\derv{t}\left[{\bf B}{\bf\cdot}\left({\bf u}+r\nabla S\right)\right]
+\nabla{\bf\cdot}\bigl\{ {\bf u}
\left[{\bf B}{\bf\cdot}
\left({\bf u}+r\nabla S\right)
\right]
+\left(h-\frac{1}{2}|{\bf u}|^2\right) {\bf B}\bigr\}=0, 
\label{eq:gch1}
\eeqn
where
\beqn
{\bf u}=\nabla\phi-r\nabla S-\tilde{\lambda}\nabla\mu 
-\frac{(\nabla\times\boldsymbol{\Gamma})\times{\bf B}}{\rho} 
-\boldsymbol{\Gamma}\frac{\nabla{\bf\cdot}{\bf B}}{\rho}, \label{eq:gch2}
\eeqn
is the Clebsch variable representation for the fluid velocity ${\bf u}$, 
and $r({\bf x},t)$ is the Lagrangian temperature integral (\ref{eq:nl11}) 
moving with the flow. 

In the special cases of either (\romannumeral1)\ ${\bf B}{\bf\cdot}\nabla S=0$ 
or (\romannumeral2)\ the case of a barotropic gas with $p=p(\rho)$, the conservation law (\ref{eq:gch1}) reduces to the usual cross helicity conservation law:
\beqn
\derv{t} ({\bf u}{\bf\cdot}{\bf B}) 
+\nabla{\bf\cdot}\left[{\bf u}({\bf u\cdot}{\bf B}) 
+\left(h-\frac{1}{2}|{\bf u}|^2\right) {\bf B}\right]=0, 
\label{eq:gch3}
\eeqn
In general the cross helicity conservation equation (\ref{eq:gch1}) is a 
nonlocal conservation law, in which the variable $r({\bf x},t)$ is a nonlocal 
potential given by (\ref{eq:nl11}).
\end{proposition}

\begin{proof}
The simplest approach is to start from the cross helicity equation 
(\ref{eq:crossh7}) with source term $T{\bf B}{\bf\cdot}\nabla S$, i.e., 
\beqn
\derv{t}\left({\bf u}{\bf\cdot}{\bf B}\right)+\nabla{\bf\cdot}\left[({\bf u}{\bf\cdot}{\bf B}){\bf u}+\left(h-\frac{1}{2} |{\bf u}|^2\right) {\bf B}\right]
=T{\bf B}{\bf\cdot}\nabla S, \label{eq:gch3a}
\eeqn
 and then 
show that the source term can be written as a pure space and time 
divergence term. First note that 
\beqn
T{\bf B}{\bf\cdot}\nabla S=\rho T\frac{{\bf B\cdot}\nabla S}{\rho}
=\rho T\psi \quad\hbox{where}\quad 
\psi=\frac{{\bf B}{\bf\cdot}\nabla S}{\rho}. 
\label{eq:gch4}
\eeqn
Note that $d\psi/dt=0$ as $\psi$ is an advected invariant. 
Using the Eulerian mass continuity equation and (\ref{eq:nl3}) 
we obtain:
\beqn
\deriv{(\rho r)}{t}+\nabla{\bf\cdot}({\bf u}\rho r)=\rho \frac{dr}{dt}
=-\rho T. \label{eq:gch5}
\eeqn
 It follows from (\ref{eq:gch4}) and (\ref{eq:gch5}) that 
$T{\bf B}{\bf\cdot}\nabla S=\rho T\psi$ and hence:
\begin{align}
T{\bf B}{\bf\cdot}\nabla S=&-\psi
\bigl\{\derv{t}(\rho r)
+\nabla{\bf\cdot}({\bf u}\rho r)\bigr\}\nonumber\\
\equiv&-\bigl\{ \derv{t}(\rho r\psi)
+\nabla{\bf\cdot}({\bf u}\rho r\psi)\big\}\nonumber\\
\equiv&-\bigl\{\derv{t}(r{\bf B}{\bf\cdot}\nabla S)
+\nabla{\bf\cdot}[{\bf u}(r{\bf B}{\bf\cdot}\nabla S)]\bigr\}. 
\label{eq:gch6}
\end{align}
Use of (\ref{eq:gch6}) in (\ref{eq:gch3a}) then gives the conservation law 
(\ref{eq:gch1}). This completes the proof.
\end{proof}

\section{Concluding Remarks}

The main aim of the present paper is to provide an overview of the 
Lie dragging  and conservation laws associated with 
fluid relabelling symmetries in MHD and fluid dynamics. Two notable new 
results in the paper are the generalization of the helicity 
conservation equation in ideal fluid mechanics and the generalization of the 
cross helicity conservation law in ideal MHD. In most derivations of 
these conservation laws it is assumed either that (\romannumeral1)\ the gas
is isentropic and the gas pressure is isobaric or (\romannumeral2)\ in the case of cross helicity conservation law in MHD $p=p(\rho,S)$ and 
${\bf B}{\bf\cdot}\nabla S=0$, meaning that the magnetic field lies in 
the $S=const.$ surfaces; in MHD the assumption $p=p(\rho)$ 
also leads to the cross-helicity conservation law. The 
assumptions (\romannumeral1) and (\romannumeral2) ensure that source terms 
dependent on $\nabla S$ 
vanish. The resultant conservation laws are local, meaning that the 
conserved density $D$ and flux ${\bf F}$ in the conservation law depend 
only on the local variables $(\rho,{\bf u}, {\bf B}, S,T)$.
A preliminary account of advected invariants in MHD using the ideas of 
Lie dragging is given in Webb et al. (2013). 
It turns out that one can obtain a nonlocal version of the fluid helicity 
conservation equation by using Clebsch variables to describe the fluid. 
In this formulation the fluid velocity is represented in terms of the 
Clebsch potentials by the formula:
\beqn
{\bf u}=\nabla\phi-r\nabla S-\tilde{\lambda} \nabla \mu, \label{eq:concl1}
\eeqn
(the usual Clebsch variables used in the variational principle are 
$\phi$, $\beta=r\rho$ and $\lambda=\tilde{\lambda}\rho$ (e.g. Zakharov 
and Kuznetsov (1997)). The resultant helicity conservation 
equation (\ref{eq:nl1}) applies for a general non-isobaric 
equation of state for the gas (i.e. 
$p=p(\rho,S)$, but also involves the nonlocal Clebsch potentials 
\begin{align}
r({\bf x},t)=&-\int_0^t T_0({\bf x}_0,t')\ dt' +r_0({\bf x}_0), \nonumber\\
\phi({\bf x},t)=&\int_0^t \ \left(\frac{1}{2} |{\bf u}|^2-h\right) 
({\bf x}_0,t')\ dt' +\phi_0({\bf x}_0) \label{eq:concl2}
\end{align}
where $T_0({\bf x}_0,t)=T({\bf x},t)$ is the temperature of the gas 
and ${\bf x}={\bf x}({\bf x}_0,t)$ is the solution of the differential 
equation system $d{\bf x}/dt={\bf u}({\bf x},t)$ with ${\bf x}={\bf x}_0$ 
at time $t=0$ (this is technically referred to as the Lagrangian map). 
Here $r_0({\bf x}_0)$ and $\phi_0({\bf x}_0)$ are integration `constants'.  
A similar non-local conservation law (\ref{eq:gch1}) for the cross helicity 
is obtained by using the Clebsch variables appropriate for MHD.

An alternative account of MHD conservation laws, Lie symmetries 
and variational methods is 
to use the Euler Poincar\'e approach 
to Noether's theorems  adopted by Cotter and Holm (2012) 
(see also Holm et al. (1998)). The Euler-Poincar\'e variational 
approach takes into account known symmetries of the Lagrangian
and uses Eulerian variations of the action.  
In the case of Lagrangian fluid mechanics the Lagrangian map 
${\bf x}={\bf x}({\bf x}_0,t)\equiv g{\bf x}_0$ can be thought of 
as a group of diffeomorphisms that map the Lagrange labels ${\bf x}_0$ 
onto the Eulerian position of the fluid element ${\bf x}$. Note 
that the group element $g$ has inverse element $g^{-1}$ 
where ${\bf x}_0=g^{-1}{\bf x}$, 
provided the Jacobian of the map is non-zero and bounded, 
and that the identity element $e$ corresponds to the transformation
${\bf x}=e{\bf x}_0={\bf x}_0$.
The use of Lie symmetries for differential equations and 
Noether's theorems are described in standard texts (e.g. Olver (1993)). 

The relationship between the helicity and cross helicity conservation  
laws for barotropic and non-barotropic equations of state for the gas, 
will be investigated in a companion paper (paper II) using Noether's theorems, 
fluid relabelling symmetries and gauge transformations. The relationship 
between the fluid relabelling symmetries and the Casimir invariants 
(e.g. Padhye and Morrison (1996a,b), Padhye (1998)) will also be  
investigated in paper II.  
 
Other approaches to conservation laws and Noether's theorems 
may be useful in future analyses. Anco and Bluman (2002) 
have developed a method to determine conservation laws 
of a system of partial differential equations that 
does not invoke Noether's theorem and a variational 
formulation of the equations 
(see also Bluman, Cheviakov and Anco (2010)). 
Noether's theorems and conservation laws using the method 
of moving frames has been developed by Goncalves and Mansfield (2012). 
This approach investigates the mathematical 
structure behind the Euler Lagrange 
equations. They give examples of variational problems that are invariant 
under semi-simple Lie algebras. The method of moving 
frames and its relation 
to Lie pseudo algebras was developed by Fels and Olver (1998).  

\section*{Aknowledgements}
GMW acknowledges stimulating discussions of MHD conservation laws with
Darryl Holm. {\bf We also acknowledge discussion of the MHD Poisson bracket 
with Phil. Morrison.}  We acknowledge a thorough report by one of the referees, 
which resulted in a much improved presentation.  
QH was supported in
part by NASA grants NNX07AO73G and NNX08AH46G. B. Dasgupta was
supported in part by an internal UAH grant. GPZ was supported
in part by NASA grants
NN05GG83G and NSF grant
nos. ATM-03-17509 and ATM-04-28880.
JFMcK acknowledges support by the NRF of South Africa.

\section*{References}

\end{document}